\documentclass[conference]{IEEEtran}
\IEEEoverridecommandlockouts
\usepackage{hyperref}
\usepackage{amsthm}
\usepackage{tikz}
\usetikzlibrary{positioning,arrows.meta}
\usepackage{cite}
\usepackage{amsmath,amssymb,amsfonts}
\usepackage{algorithm}
\usepackage{graphicx}
\usepackage{textcomp}
\usepackage{xcolor}
\usepackage{graphicx}
\usepackage{algpseudocode}
\usepackage{booktabs}
\usepackage{tabularx}
\numberwithin{equation}{section}
\numberwithin{figure}{section}
\numberwithin{table}{section}
\newtheorem{theorem}{Theorem}[section]
\newtheorem{lemma}{Lemma}[section]
\newtheorem{corollary}{Corollary}[section]

\newtheorem{definition}{Definition}[section]

\newtheorem{remark}{Remark}[section]

\def\BibTeX{{\rm B\kern-.05em{\sc i\kern-.025em b}\kern-.08em
    T\kern-.1667em\lower.7ex\hbox{E}\kern-.125emX}}
\begin{document}

\title{Target-Stratified Fair Range Summaries: Improved Fair $\varepsilon$-Nets
	and Geometric Hitting Sets\\
}

\author{
	\IEEEauthorblockN{
		Mingchao Zhou\IEEEauthorrefmark{1},
		Lei Zhao\IEEEauthorrefmark{1},
		Zhipeng Cai\IEEEauthorrefmark{2},
		Zhao Zhang\IEEEauthorrefmark{1}\textsuperscript{*}
	}
	
	\IEEEauthorblockA{
		\IEEEauthorrefmark{1}
		\textit{School of Mathematical Sciences, Zhejiang Normal University,}\\
		Jinhua 321004, Zhejiang, China\\
		Emails: zmc138575@gmail.com, zhaolei@zjnu.edu.cn, hxhzz@sina.com
	}
	
	\IEEEauthorblockA{
		\IEEEauthorrefmark{2}
		\textit{Department of Computer Science, Georgia State University,}\\
		Atlanta, GA 30303, USA\\
		Email: zcai@gsu.edu
	}
	
	\IEEEauthorblockA{
		ORCID:
		Mingchao Zhou:
		\href{https://orcid.org/0009-0007-3789-9872}{0009-0007-3789-9872};
		Lei Zhao:
		\href{https://orcid.org/0009-0005-9682-7284}{0009-0005-9682-7284}\\
		Zhipeng Cai:
		\href{https://orcid.org/0000-0001-6017-975X}{0000-0001-6017-975X};
		Zhao Zhang:
		\href{https://orcid.org/0000-0003-4191-7598}{0000-0003-4191-7598}
	}
	
	\thanks{\textsuperscript{*}Corresponding author: Zhao Zhang.}
}
\maketitle

\begin{abstract}
Compact summaries are a key tool for approximate query processing over large
datasets. For range-query workloads, an $\varepsilon$-net provides a small
summary that hits every sufficiently large range. However, classical
$\varepsilon$-nets only guarantee range validity and do not control the group
composition of the selected tuples. As a result, the summary may be
range-valid but poorly representative, which can propagate imbalance to
downstream query results.

Motivated by recent work on fair $\varepsilon$-nets and fair geometric hitting
sets \cite{dehghankar2025fair}, we study fairness-aware range summaries under prescribed target group
ratios.  Different from previous sample-and-repair approach, we propose a target-stratified sampling method.  For demographic parity (in which the ratio of fairness is determined by group proportion), our sample size is $O(A_{\varepsilon})$, coinciding with the standard 
$\varepsilon$-net bound, improving previous bound of $O\!\left(A_\varepsilon\log\frac{k}{\varphi}\right)$. For
custom-ratio targets (in which the ratio of fairness is determined by manually defined proportion), our sample size is $O(A_{\Gamma})$, where $\Gamma$ is a parameter measuring the gap between the customized ratio and the demographic parity; we prove that this dependence on $\Gamma$ is unavoidable, with a
worst-case lower bound of $\Omega(\Gamma/\varepsilon)$. Using our target-stratified sampling method, we could improve the previous approximation ratio for the fair geometric hitting set problem by a logarithmic factor, and making use of this result, we could in turn improve the size of custom-ratio fair $\varepsilon$-net. Experiments on real and synthetic datasets demonstrate that our method
constructs smaller fair summaries than existing approaches, scales to large
datasets and fine-grained group constraints, and improves downstream range
query processing.
\end{abstract}


\section{Introduction}

Compact summaries are a standard tool for approximate query processing over
large datasets \cite{cormode2011synopses,li2018approximate}. 
In approximate query processing, the goal is to return approximate answers
efficiently when exact query evaluation is too costly for interactive data
analysis \cite{li2018approximate,agarwal2013blinkdb}. 
A long line of work has developed synopsis-based approaches for this purpose.  For example,
Aqua \cite{acharya1999aqua} uses precomputed synopses to provide fast approximate answers to
aggregate queries in data-warehousing and OLAP applications, while samples, histograms, wavelets, and sketches form
standard synopsis families for massive data \cite{cormode2011synopses}. 
More recently, systems such as BlinkDB \cite{agarwal2013blinkdb} demonstrate that sampling-based
summaries can support ad-hoc, interactive SQL queries over large volumes of
data with explicit accuracy--latency tradeoffs.
These works motivate the use of compact summaries for reducing query
cost while retaining useful accuracy guarantees
\cite{chaudhuri2017approximate,li2018approximate}.

Range-query workloads form a basic class of analytical queries in spatial,
geometric, and multidimensional data. A range query asks for records lying in
a query region, such as an interval, rectangle, halfspace, disk, or another
predicate-defined range; range-searching data structures have been studied
extensively for efficiently reporting or counting such records
\cite{agarwal1999geometric}. For these workloads, a classical and
mathematically clean summary model is an $\varepsilon$-net. Given a range
space $(X,\mathcal R)$, an $\varepsilon$-net is a subset of the input that
intersects every range containing at least an $\varepsilon$ fraction of the
data. Thus, an $\varepsilon$-net preserves all sufficiently large query
regions. For range spaces of bounded VC dimension, the size of an
$\varepsilon$-net can be bounded in terms of the range-space complexity and
the accuracy parameter, rather than the input size \cite{HausslerWelzl1987}.
This makes $\varepsilon$-nets a natural abstraction for compact summaries.

However, range validity alone does not guarantee representation fairness
\cite{shahbazi2023representation}. In many data-management applications,
tuples are associated with demographic or semantic groups, such as gender,
race, region, or item category \cite{shetiya2022fairness}. A classical
$\varepsilon$-net may hit every heavy range while still overrepresenting some
groups and underrepresenting others. Figure~\ref{fig:intro-fairness} illustrates this gap between range validity and
representation fairness. Here, the blue rectangle denotes a heavy range, and
the green-circled points form a selected summary. The left panel shows that
a summary may be range-valid, in the sense that it hits the heavy range, while
still being unfair because it overrepresents the red group. The right panel shows
a summary that satisfies both goals simultaneously: it hits the heavy range
and respects the intended group proportions. Hence, fairness is not implied by
validity alone and must be enforced explicitly. 

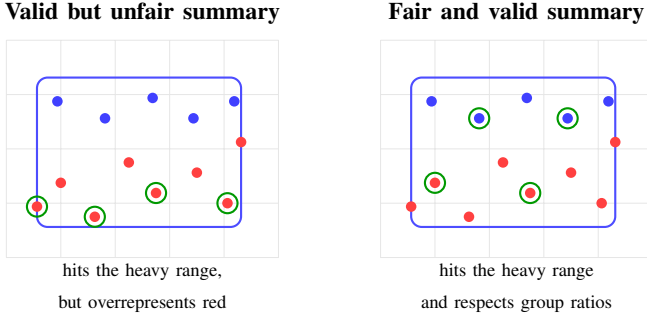
\begin{figure}[htbp]
	\centering
	\begin{tikzpicture}[scale=0.9]
		\begin{scope}
			\node at (2,3.6) {\small \textbf{Valid but unfair summary}};
			\draw[step=0.8,gray!20,very thin] (0,0) grid (4,3.2);
			\draw[thick,rounded corners,blue!70] (0.45,0.45) rectangle (3.45,2.65);
			
			\foreach \x/\y in {
				0.45/0.75,0.8/1.1,1.3/0.6,1.8/1.4,2.2/0.95,2.8/1.25,3.25/0.8,3.45/1.7}
			\fill[red!75] (\x,\y) circle (2.2pt);
			
			\foreach \x/\y in {
				0.75/2.3,1.45/2.05,2.15/2.35,2.75/2.05,3.35/2.3}
			\fill[blue!75] (\x,\y) circle (2.2pt);
			
			\foreach \x/\y in {0.45/0.75,1.3/0.6,2.2/0.95,3.25/0.8}
			\draw[green!60!black,thick] (\x,\y) circle (4.2pt);
			
			\node[align=center] at (2,-0.45) {\scriptsize hits the heavy range,\\ \scriptsize but overrepresents red};
		\end{scope}
		
		\begin{scope}[xshift=5.5cm]
			\node at (2,3.6) {\small \textbf{Fair and valid summary}};
			\draw[step=0.8,gray!20,very thin] (0,0) grid (4,3.2);
			\draw[thick,rounded corners,blue!70] (0.45,0.45) rectangle (3.45,2.65);
			
			\foreach \x/\y in {
				0.45/0.75,0.8/1.1,1.3/0.6,1.8/1.4,2.2/0.95,2.8/1.25,3.25/0.8,3.45/1.7}
			\fill[red!75] (\x,\y) circle (2.2pt);
			
			\foreach \x/\y in {
				0.75/2.3,1.45/2.05,2.15/2.35,2.75/2.05,3.35/2.3}
			\fill[blue!75] (\x,\y) circle (2.2pt);
			
			\foreach \x/\y in {0.8/1.1,2.2/0.95}
			\draw[green!60!black,thick] (\x,\y) circle (4.2pt);
			\foreach \x/\y in {1.45/2.05,2.75/2.05}
			\draw[green!60!black,thick] (\x,\y) circle (4.2pt);
			
			\node[align=center] at (2,-0.45) {\scriptsize hits the heavy range\\ \scriptsize and respects group ratios};
		\end{scope}
	\end{tikzpicture}
	\caption{Range validity alone does not guarantee representation fairness. A classical $\varepsilon$-net may hit every heavy range while still overrepresenting one group.}
	\label{fig:intro-fairness}
\end{figure}

If such a summary is used to answer range queries, display representative
records, or select items for downstream analysis, the imbalance in the summary
may be inherited by the query output. This concern is consistent with recent
work viewing fairness as a data-management problem: unfairness may originate
from the underlying data and from data-processing operations such as selection,
repair, or exposure \cite{salimi2019data,salimi2019interventional,shetiya2022fairness}.
Recently, fair ranking and retrieval methods explicitly control group
representation in displayed top-$k$ results
\cite{zehlike2017fa,singh2018fairness,celis2018ranking}. These observations
motivate fairness-aware summaries that are simultaneously range-valid and
group-representative.

Fair $\varepsilon$-nets and fair geometric hitting sets were recently
introduced by Dehghankar et al.~\cite{dehghankar2025fair}. Their work adds
group-fairness constraints to classical geometric approximation problems,
including $\varepsilon$-nets, $\varepsilon$-samples, and geometric hitting
sets. They study both demographic parity, where the output preserves the input
group proportions, and custom-ratio fairness, where the target proportions are
specified externally. They also develop multiple algorithmic approaches  to construct fair $\varepsilon$-nets,
including a fast sampling-based method and a discrepancy-based method, and
show how fair geometric hitting set can be  solved using fair $\varepsilon$-nets.

The construction method in~\cite{dehghankar2025fair} could be called sample-and-repair: it first constructs a random sample and then repairs the group composition by adding points from
underrepresented groups.  The repair step increases the size of the $\varepsilon$-net by a logarithmic
factor, compared with the standard unfair $\varepsilon$-net. A natural question arises: 

\begin{quote}
	Is it possible to remove the additional logarithmic factor by some other method?
\end{quote}

\subsection{Contributions}

Our main contributions are as follows.

\begin{itemize}
\item  We propose a target-stratified sampling method for fair range summaries. Our method fixes the sampling size for each group according to the prescribed target ratios and then samples within each group. As a consequence, fairness is enforced by construction, and the remaining task is to certify the range-hitting property. We prove that by this method, the standard $\varepsilon$-net sample-size guarantee can be recovered for demographic parity.
	
\item For demographic-parity fair $\varepsilon$-nets,  we obtain bound
	$O\!\left(\frac{1}{\varepsilon}\left(d\log\frac{1}{\varepsilon}
	+\log\frac{1}{\varphi}\right)\right)$ in range spaces of VC dimension $d$,
	improving the previous sample-and-repair bound
	$O\!\left(\frac{1}{\varepsilon}\left(d\log\frac{1}{\varepsilon}
	+\log\frac{1}{\varphi}\right)\log\frac{k}{\varphi}\right)$
	\cite{dehghankar2025fair}, where $\varphi$ is a parameter controlling the failure probability.
	
	\item For custom-ratio fair $\varepsilon$-nets, we introduce a
	distribution-shift parameter $\Gamma$ to measure the gap between the customized ratio and the demographic parity. We obtain bound
	$O\!\left(\frac{\Gamma}{\varepsilon}\left(\log m+\log\frac{1}{\varphi}\right)\right)$
	for finite range families of size $m$, and bound
	$O\!\left(\frac{\Gamma}{\varepsilon}\left(d\log\frac{\Gamma}{\varepsilon}
	+\log\frac{1}{\varphi}\right)\right)$ for range spaces of VC dimension $d$.
	We also prove that the leading term $\Gamma/\varepsilon$ is unavoidable
	for  custom-ratio fairness.
	
 \item 
 We improve the approximation ratio for Fair Geometric Hitting Set
(FGHS) from
$O\!\left(L_{\mathrm{FGHS}}\log\frac{k}{\varphi}\right)$ obtained
in~\cite{dehghankar2025fair} to $O(L_{\mathrm{FGHS}})$, where
$L_{\mathrm{FGHS}}=d\log \mathrm{OPT}_{\mathrm{FGHS}}+\log(1/\varphi)$,
$\mathrm{OPT}_{\mathrm{FGHS}}$ is the optimum size of the FGHS instance.
Interestingly, we could reduce custom-ratio fair
$\varepsilon$-nets to FGHS and obtain an $O(L_{\mathrm{CR}})$ approximation
ratio, where
$L_{\mathrm{CR}}=d\log \mathrm{OPT}_{\mathrm{CR}}+\log(1/\varphi)$ and
$\mathrm{OPT}_{\mathrm{CR}}$ is the optimum size of a custom-ratio fair
$\varepsilon$-net. This improves the previous bound of
$O\!\left(L_{\mathrm{CR}}\log\frac{k}{\varphi}\right)$ \cite{dehghankar2025fair}.

\item We conduct experiments on real and synthetic datasets.
The results show that target-stratified sampling produces smaller fair
summaries than existing methods, scales to large datasets and fine-grained
group partitions, and follows the predicted dependence on $\Gamma$.
We also evaluate our rounding method for FGHS and show that the resulting
fair summaries can improve downstream range-query filtering.
\end{itemize}

A detailed comparison between our results and the framework of
\cite{dehghankar2025fair} is provided in Table~\ref{tab:main-comparison}.

\begin{table*}[t]
	\centering
	\caption{Comparison with the framework of~\cite{dehghankar2025fair}.}
	\label{tab:main-comparison}
	\footnotesize
	\renewcommand{\arraystretch}{1.22}
	\setlength{\tabcolsep}{3.5pt}
	\begin{tabular}{lcccc}
		\toprule
		Setting
		& Previous size
		& Our size
		& Previous time
		& Our time \\
		\midrule
		
		DP fair $\varepsilon$-net, sampling
		&
		$O\!\left(A_\varepsilon\log\frac{k}{\varphi}\right)$
		&
		$O(A_\varepsilon)$
		&
		$O(n)$
		&
		$O(n)$
		\\
		
		DP fair $\varepsilon$-net, discrepancy
		&
		$O(D_\varepsilon)$
		&
		$O(A_\varepsilon)$
		&
		$O(nm\log n)$
		&
		$O(n)$
		\\
		
		DP fair $\varepsilon$-net, sketch-and-merge
		&
		$O(D_\varepsilon)$
		&
		$O(A_\varepsilon)$
		&
		$\displaystyle
		O\!\left(
		n\cdot
		\frac{d^{3d}}{\varepsilon^{2d}}
		\log^d\frac{d}{\varepsilon}
		\right)$
		&
		$O(n)$
		\\
		
		CR fair $\varepsilon$-net, sampling
		&
		--
		&
		$O(A_\Gamma)$
		&
		--
		&
		$O(n)$
		\\
		
		CR fair $\varepsilon$-net, sampling lower bound
		&
		--
		&
		$\Omega\!\left(\frac{\Gamma}{\varepsilon}\right)$
		&
		N/A
		&
		N/A
		\\
		
		Fair Geometric Hitting Set
		&
		$O\!\left(
		OPT_{\mathrm{FGHS}}L_{\mathrm{FGHS}}
		\log\frac{k}{\varphi}
		\right)$
		&
		$O\!\left(
		OPT_{\mathrm{FGHS}}L_{\mathrm{FGHS}}
		\right)$
		&
		$T_{\mathrm{LP}}(m)$
		&
		$T_{\mathrm{LP}}(m)$
		\\
		
		CR fair $\varepsilon$-net via FGHS
		&
		$O\!\left(
		OPT_{\mathrm{CR}}L_{\mathrm{CR}}
		\log\frac{k}{\varphi}
		\right)$
		&
		$O\!\left(
		OPT_{\mathrm{CR}}L_{\mathrm{CR}}
		\right)$
		&
		$T_{\mathrm{LP}}(m_\varepsilon)$
		&
		$T_{\mathrm{LP}}(m_\varepsilon)$
		\\
		
		\bottomrule
	\end{tabular}
	
	\vspace{1mm}
	\begin{minipage}{0.98\textwidth}
		\footnotesize
		Here
		$A_\varepsilon=
		\frac1\varepsilon
		\left(d\log\frac1\varepsilon+\log\frac1\varphi\right)$,
		$A_\Gamma=
		\frac{\Gamma}{\varepsilon}
		\left(d\log\frac{\Gamma}{\varepsilon}+\log\frac1\varphi\right)$,
		$D_\varepsilon=
		\frac{d}{\varepsilon}\log\frac{d}{\varepsilon}$,
		$L_{\mathrm{FGHS}}=
		d\log OPT_{\mathrm{FGHS}}+\log\frac1\varphi$, and
		$L_{\mathrm{CR}}=
		d\log OPT_{\mathrm{CR}}+\log\frac1\varphi$.
		Also, $m=|\mathcal R|$, $m_\varepsilon=|\mathcal R_\varepsilon|$, and
		$T_{\mathrm{LP}}(m)=O(nm+M(n,m+k))$.
		The lower-bound row is not an algorithmic result; it shows that the
		leading $\Gamma/\varepsilon$ dependence is unavoidable for direct
		custom-ratio constructions.
	The discrepancy and sketch-and-merge rows are included to position our result
	within the previous framework. Their bounds are based on discrepancy-halving
	constructions and are stated under standard halving assumptions, such as
	power-of-two input and color-class sizes and nonzero target counts for all
	groups. Our contribution is a direct linear-time sampling rule, requiring only
	target-count feasibility.
	\end{minipage}
\end{table*}

\subsection{Related Work}
\label{sec:related}

Classical $\varepsilon$-nets provide small representative subsets for range
spaces of bounded VC dimension. Haussler and Welzl~\cite{HausslerWelzl1987}
established the classical sampling result for $\varepsilon$-nets.
Br\"onnimann and Goodrich~\cite{bronnimann1994almost} further showed how
$\varepsilon$-net constructions can be used to obtain approximation algorithms
for set cover and hitting set in finite-VC systems. These results provide the classical geometric and algorithmic foundations for the problems studied in this paper.

Fairness has also been studied in several related data-selection and
summarization problems. Celis et al.~\cite{celis2018fair} studied fair and
diverse data summarization using determinantal point processes, where fairness
controls the representation of different groups in the selected summary.
Huang et al.~\cite{huang2019coresets} developed coresets for clustering under
fairness constraints. In the database setting, Shetiya et
al.~\cite{shetiya2022fairness} studied fairness-aware range queries, where the
goal is to find a fair range that is close to a user-specified query.
These works consider different objectives from ours: our goal is to construct
a reusable fair summary that intersects every sufficiently large range.
Fairness has also been incorporated into covering and hitting problems.
Bandyapadhyay et al.~\cite{bandyapadhyay2021fair} studied fair covering and
hitting problems in which group-specific coverage requirements must be
satisfied. Inamdar et al.~\cite{inamdar2025fixed} studied a Fair Hitting Set
problem in which the selected hitting set is required not to contain too many
elements of each group. Dehghankar et al.~\cite{dehghankar}
introduced Fair Set Cover, where the selected sets are required to satisfy
demographic parity. These formulations differ from ours, which requires the
selected points themselves to satisfy prescribed group ratios while hitting
all sufficiently large geometric ranges.

The work most closely related to ours is that of Dehghankar, Sintos, and
Asudeh~\cite{dehghankar2025fair}, who introduced fair $\varepsilon$-nets,
fair $\varepsilon$-samples, and Fair Geometric Hitting Set under demographic
parity (DP) and custom-ratio (CR) fairness. Their sampling-based method first
draws a random sample from the whole dataset and then adds points from
individual groups to restore the prescribed group ratios. This leads to an
additional $\log(k/\phi)$ factor in the high-probability sample-size bound.
They also develop discrepancy-based constructions for fair $\varepsilon$-nets
and an LP-based approach for Fair Geometric Hitting Set.

\subsection{Organization}

The rest of the paper is organized as follows.
Section~\ref{sec:preliminaries} introduces the notation and definitions used
throughout the paper.
Section~\ref{sec:target-stratified-fair-eps-net} presents the target-stratified
sampling framework for fair $\varepsilon$-nets.
In Section~\ref{Target-Stratified Sampling}, we describe the sampling rule.
Section~\ref{Heavy Range} establishes the connection between
$\varepsilon$-heavy ranges under the uniform distribution and heavy ranges under
the target-stratified distribution.
Section~\ref{sec:finite-range-family} gives the guarantee for finite range
families.
Section~\ref{VC-Bounded } proves the stratified VC $\varepsilon$-net guarantee
for samples with fixed color counts.
Section~\ref{Consequences for DP} derives consequences for demographic
parity and custom-ratio fairness, including the dependence on the
distribution-shift parameter $\Gamma$.
Section~\ref{sc improve fair ge hit} develops an improved LP-rounding approximation
for Fair Geometric Hitting Set.
Section~\ref{from fghs to custom} reduces the size of custom-ratio fair $\varepsilon$-net using fair geometric hitting set. 
Section~\ref{sec:experiments} reports the experimental evaluation, and
Section~\ref{sec:conclusion} concludes the paper.

\section{Preliminaries}
\label{sec:preliminaries}

We introduce the notation and definitions used throughout the paper.
For an integer $k$, let $[k]=\{1,\ldots,k\}$. 

\begin{definition}[Range space~\cite{HausslerWelzl1987,Matousek2002}]
	A \emph{range space} is a pair $(X,\mathcal R)$, where $X$ is a finite
	ground set and $\mathcal R\subseteq 2^X$ is a family of subsets of $X$,
	called ranges. We write $n=|X|$.
\end{definition}

In geometric applications, the ground set $X$ consists of points, and the
ranges in $\mathcal R$ are induced by geometric objects such as rectangles,
disks, or halfspaces. In data-management applications, points in $X$ may
represent database tuples, and a range $R\in\mathcal R$ may represent the
answer set of a range query.

\begin{definition}[$\varepsilon$-net~\cite{HausslerWelzl1987,Matousek2002}]
	A set $S\subseteq X$ is an $\varepsilon$-net for $(X,\mathcal R)$ if
	\[
	\forall R\in\mathcal R \ \mbox{with}\ |R|\ge \varepsilon |X|
	\quad\Longrightarrow\quad
	S\cap R\ne\emptyset.
	\]
	More generally, let $q$ be a probability distribution on $X$ and define
	\[
	q(R)=\sum_{p\in R}q(p).
	\]
	A set $S\subseteq X$ is an $\alpha$-net with respect to $q$ if
	\[
	\forall R\in\mathcal R \ \mbox{with}\ q(R)\ge \alpha
	\quad\Longrightarrow\quad
	S\cap R\ne\emptyset.
	\]
\end{definition}

The second definition is the distributional version of an $\varepsilon$-net.
The unweighted definition is recovered by taking $q$ to be the uniform
distribution on $X$.

\begin{definition}[Trace and VC dimension~\cite{VapnikChervonenkis1971,Matousek2002}]
	For a subset $Y\subseteq X$, the \emph{trace} of $\mathcal R$ on $Y$ is
	\[
	\mathcal R|_Y
	=
	\{R\cap Y:\ R\in\mathcal R\}.
	\]
	The VC dimension of $(X,\mathcal R)$ is the largest integer $d$ such that
	there exists a subset $Y\subseteq X$ with $|Y|=d$ and
	\[
	\mathcal R|_Y=2^Y.
	\]
\end{definition}

The following result is known as Sauer--Shelah lemma~\cite{Sauer1972,Shelah1972}: if
$(X,\mathcal R)$ has VC dimension $d$, then for every finite set
$Y\subseteq X$,
\[
|\mathcal R|_Y|
\le
\left(\frac{e|Y|}{d}\right)^d .
\]
This lemma bounds the number of distinct intersections induced by
$\mathcal R$ on a finite sample.

\begin{theorem}[Standard VC $\varepsilon$-net theorem~\cite{VapnikChervonenkis1971,HausslerWelzl1987}]
	\label{thm:classical-vc-eps-net}
	Let $(X,\mathcal R)$ be a range space with VC dimension $d$. For any
	$\alpha\in(0,1)$ and $\varphi\in(0,1)$, a multiset $S$ obtained by taking
	\[
	O\left(
	\frac{1}{\alpha}
	\left(
	d\log\frac{1}{\alpha}
	+
	\log\frac{1}{\varphi}
	\right)
	\right)
	\]
	independent samples from $X$ uniformly at random (with replacement),
	is an $\alpha$-net with probability at least $1-\varphi$.
\end{theorem}

We now introduce the fairness constraints. The ground set is partitioned into
$k$ color classes:
\[
X=X_1\cup X_2\cup\cdots\cup X_k .
\]
Let
\[
\mu_c=\frac{|X_c|}{|X|}
\]
be the input proportion of color class $X_c$. A target ratio vector is a vector
\[
\tau=(\tau_1,\ldots,\tau_k),
\qquad
\tau_c\ge 0,
\qquad
\sum_{c=1}^k \tau_c=1.
\]

\begin{definition}[Fair $\varepsilon$-net~\cite{dehghankar2025fair}]
	A set $S\subseteq X$ is a \emph{fair $\varepsilon$-net} with target ratio vector 
	$\tau$ if $S$ is an $\varepsilon$-net and
	\[
	|S\cap X_c|
	=
	\tau_c |S|,
	\qquad
	\forall c\in[k].
	\]
\end{definition}

Two choices of $\tau$ will be used \cite{dehghankar2025fair}. Under \emph{demographic parity}, the target ratios are equal to the input
proportions:
\[
\tau_c=\mu_c,
\qquad
\forall c\in[k].
\]
Under \emph{custom-ratio fairness}, the target vector
$\tau$ is specified independently of the input proportions.

Fair $\varepsilon$-nets and related fair geometric hitting set problem
were studied by Dehghankar et al.~\cite{dehghankar2025fair}. Their
sampling-based algorithm first draws a random sample and then adds points to
restore the prescribed color ratios. In contrast, our target-stratified
sampling rule fixes the color counts before sampling.

To measure the bias of custom-ratio from the uniform distribution on
$X$, define
\[
\Gamma
=
\max_{c:\tau_c>0}
\frac{\mu_c}{\tau_c}.
\]
If there exists a color class $X_c$ with $\mu_c>0$ and $\tau_c=0$, then we set
$\Gamma=+\infty$. In this case, the target ratios forbid selecting points from
$X_c$. Consequently, if an $\varepsilon$-heavy range is contained in such a
color class, then no set satisfying the target ratios can be an $\varepsilon$-net.

Throughout the paper, we adopt the same feasibility convention as
\cite{dehghankar2025fair}. First, we assume the following support-compatibility
condition:
\[
\mu_c>0 \quad \Longrightarrow \quad \tau_c>0,
\qquad \forall c\in[k].
\]
Under this convention, $\Gamma<+\infty$.
This convention is necessary for obtaining a nontrivial guarantee under exact
target-ratio constraints. In addition, whenever a sample size $\lambda$ is used,
we require the prescribed target counts
\[
\lambda_c:=\tau_c\lambda
\]
to be integral and to satisfy
\[
\lambda_c\le |X_c|,\qquad \forall c\in[k].
\]
This convention is necessary for enforcing the target
ratios exactly, and will not lose much if it is not satisfied  and corresponding values are rounded to nearest integers.

\begin{definition}[Fair Geometric Hitting Set~\cite{dehghankar2025fair}]
	Given a range space $(X,\mathcal R)$ and a target ratio vector  $\tau$, the
	\emph{Fair Geometric Hitting Set} problem asks for a minimum-size set
	$H\subseteq X$ such that
	\[
	H\cap R\ne\emptyset,
	\qquad
	\forall R\in\mathcal R,
	\]
	and
	\[
	|H\cap X_c|
	=
	\tau_c |H|,
	\qquad
	\forall c\in[k].
	\]
	We denote its optimum value by $OPT_{\mathrm{FGHS}}$.
\end{definition}

The following  Chernoff lower-tail bound
\cite{mitzenmacher2017probability} will be used.  Let $Z$ be a sum of independent Bernoulli
random variables with mean $\mathbb E[Z]$. Then, for any $\delta\in(0,1)$,
\[
\Pr\left[
Z<(1-\delta)\mathbb E[Z]
\right]
\le
\exp\left(
-\frac{\delta^2\mathbb E[Z]}{2}
\right).
\]

\section{Target-Stratified Fair $\varepsilon$-Nets}
\label{sec:target-stratified-fair-eps-net}

We present a target-stratified sampling framework. 

\subsection{Target-Stratified Sampling}
\label{Target-Stratified Sampling}

\begin{algorithm}[htbp]
	\caption{\textsc{Target-Stratified-Sampling}}
	\label{alg:target-stratified-sampling}
	\begin{minipage}{0.95\linewidth}
		\textbf{Input:} range space $(X,\mathcal R)$, color partition
		$X=X_1\cup\cdots\cup X_k$, target ratio vector
		$\tau=(\tau_1,\ldots,\tau_k)$, and sample size $\lambda$.
		
		\textbf{Output:} a fair sample $S\subseteq X$.
		
		\begin{enumerate}
			\item For each color $c\in[k]$, set $\lambda_c=\tau_c\lambda$.
			\item Sample $\lambda_c$ points from $X_c$ uniformly with replacement,
			and let $M_c$ be the resulting multiset.
			\item Let $S_c$ be the set of distinct points in $M_c$.
			If $|S_c|<\lambda_c$, then add $\lambda_c-|S_c|$ arbitrary points from $X_c\setminus S_c$.
			\item Return $S=\bigcup_{c=1}^k S_c$.
		\end{enumerate}
	\end{minipage}
\end{algorithm}

The analysis is applied to the sampled multiset $M=\bigcup_c M_c$. The final
post-processing only removes duplicate copies and adds points within the same
color class. Hence it restores the prescribed color counts and preserves every
range already hit by $M$.
By our construction,
\[
|S|=\sum_{c=1}^k\lambda_c=\lambda,
\qquad
\frac{|S\cap X_c|}{|S|}
=
\frac{\lambda_c}{\lambda}
=
\tau_c,
\qquad
\forall c\in[k].
\]
Thus the target-ratio constraints hold deterministically. It remains to prove
that $S$ hits every $\varepsilon$-heavy range.

\subsection{Heavy Ranges under the Target-Stratified Distribution}
\label{Heavy Range}
Given target ratio vector $\tau$, define the target-stratified distribution $q$ on
$X$ as follows. First choose a color class $X_c$ with probability $\tau_c$, and
then choose a point uniformly from $X_c$. Equivalently,
\[
q(p)=\frac{\tau_c}{|X_c|},
\qquad
\forall p\in X_c .
\]
For any range $R\in\mathcal R$,
\[
q(R)
=
\sum_{c=1}^k
\tau_c\frac{|R\cap X_c|}{|X_c|}.
\]
The quantity $q(R)$ is the probability that one draw a point from the
target-stratified distribution that falls in $R$.

The next lemma is a bridge between $\varepsilon$-heavy ranges and the
target-stratified sampling distribution.

\begin{lemma}
	\label{lem:heavy-under-q}
	If $R\in\mathcal R$ satisfies
	\[
	|R|\ge \varepsilon |X|,
	\]
	then
	\[
	q(R)\ge \frac{\varepsilon}{\Gamma}.
	\]
\end{lemma}

\begin{proof}
	Since $|X_c|=\mu_c|X|$, we have
	\[
	q(R)
	=
	\sum_{c=1}^k
	\tau_c\frac{|R\cap X_c|}{|X_c|}
	=
	\sum_{c=1}^k
	\frac{\tau_c}{\mu_c}
	\frac{|R\cap X_c|}{|X|}.
	\]
	By the definition of $\Gamma$, $\tau_c/\mu_c\ge 1/\Gamma$ for every
	$c\in[k]$. Therefore,
	\[
	q(R)
	\ge
	\frac1\Gamma
	\sum_{c=1}^k
	\frac{|R\cap X_c|}{|X|}
	=
	\frac{|R|}{\Gamma |X|}
	\ge
	\frac{\varepsilon}{\Gamma}.
	\]
The lemma  is proved.
\end{proof}

\subsection{Finite Range Families}
\label{sec:finite-range-family}

We first prove the guarantee for finite range families. Suppose $|\mathcal R|=m$.

\begin{theorem}[Finite range $\varepsilon$-net]
	\label{thm:finite-range-stratified}
	If the sample size
	\[
	\lambda
	\ge
	\frac{\Gamma}{\varepsilon}
	\left(
	\ln m+\ln\frac1\varphi
	\right),
	\]
	then Algorithm~\ref{alg:target-stratified-sampling} returns a fair $\varepsilon$-net with probability
	at least $1-\varphi$.
\end{theorem}

\begin{proof}
	The target-ratio constraints hold by our construction. It remains to prove that
	the output set $S$ hits every $\varepsilon$-heavy range. 	
    Note that $S$ is obtained from $M=\bigcup_{c=1}^k M_c$ (the multiset of sampled points) by removing duplicate samples. If $M$ hits a range then $S$ also hits it. Hence it suffices
	to bound the probability that $M$ misses a range.
	
	Fix a range $R\in\mathcal R$ and write
	\[
	\rho_c(R)=\frac{|R\cap X_c|}{|X_c|}.
	\]
	For color $c$, the probability that all $\lambda_c$ sampled points avoid
	$R$ is $(1-\rho_c(R))^{\lambda_c}$. Since the samples are independent across
	color classes, we have
	\[
	\begin{aligned}
		\Pr[M\cap R=\emptyset]
		&=
		\prod_{c=1}^k
		(1-\rho_c(R))^{\lambda_c}  \\
		&\le
		\exp\left(
		-\sum_{c=1}^k \lambda_c\rho_c(R)
		\right) \\
		&=
		\exp\left(
		-\lambda
		\sum_{c=1}^k
		\tau_c\frac{|R\cap X_c|}{|X_c|}
		\right) \\
		&=
		\exp(-\lambda q(R)).
	\end{aligned}
	\]
	If $R$ is $\varepsilon$-heavy, Lemma~\ref{lem:heavy-under-q} gives
	$q(R)\ge\varepsilon/\Gamma$. Thus
	\[
	\Pr[M\cap R=\emptyset]
	\le
	\exp\left(-\lambda\frac{\varepsilon}{\Gamma}\right)
	\le
	\frac{\varphi}{m}.
	\]
	A union bound over all ranges in $\mathcal R$ shows that the probability
	that some $\varepsilon$-heavy range is missed is at most $\varphi$. The theorem is proved.
\end{proof}

\subsection{VC-Bounded Range Spaces}
\label{VC-Bounded }

The next lemma extends the standard VC $\varepsilon$-net guarantee to
stratified sampling with fixed color counts. We state it for a general
distribution $q$ with $q(X_c)=\tau_c$, which includes uniform sampling within
colors as a special case.

\begin{lemma}[Stratified VC $\varepsilon$-net]
	\label{lem:stratified-vc-eps-net}
	Let $(X,\mathcal R)$ be a range space with VC dimension $d\ge 1$, where
	$X=X_1\cup\cdots\cup X_k$. Let $q$ be a distribution on $X$ such that $q(X_c)=\tau_c,~ \forall c\in[k].$
	Let $M$ be the multiset obtained by sampling, for each color $c$ with
	$\tau_c>0$, $\lambda_c$ points independently from $X_c$ according to the
	distribution $\Pr[p]=\frac{q(p)}{\tau_c}$ for $p\in X_c.$
	Then there exists a constant $C>0$ such that, for any
	$\alpha,\varphi\in(0,1)$, if
	\[
	\lambda
	\ge
	\frac{C}{\alpha}
	\left(
	d\log\frac1\alpha+\log\frac1\varphi
	\right),
	\]
	then, with probability at least $1-\varphi$, $M$ hits every range
	$R\in\mathcal R$ with $q(R)\ge\alpha$.
\end{lemma}

\begin{proof}
	Define the bad event
	\[
	\mathcal B
	=
	\left\{
	\exists R\in\mathcal R:\ q(R)\ge\alpha
	\ \text{and}\
	M\cap R=\emptyset
	\right\}.
	\]
	We prove that $\Pr[\mathcal B]\le\varphi$.
	
Draw an independent ghost multiset $T$ using the same stratified sampling rule
as $M$. For a range
$R\in\mathcal R$, let
\[
Z_R:=|T\cap R|.
\]

{\bf Claim 1.} If sample size $\lambda\ge 8\ln 2/\alpha$, then for any range $R\in\mathcal R$ with $q(R)\geq \alpha$, $\Pr\left[Z_R\ge \frac{\alpha\lambda}{2}\right]\ge
	\frac12$.
	
	[Proof of Claim 1] For any $R\in\mathcal R$ with $q(R)\geq \alpha$,

\[
\begin{aligned}
	\mathbb E[Z_R]
	&=
	\sum_{p\in R}\Pr[p\in T]\\
	&=
	\sum_{c:\tau_c>0}\sum_{p\in R\cap X_c}\Pr[p\in T]\\
	&=
	\sum_{c:\tau_c>0}
	\sum_{p\in R\cap X_c}\lambda_c\frac{q(p)}{\tau_c}\\
	&=
	\lambda
	\sum_{c:\tau_c>0}q(R\cap X_c)\\
	&=
	\lambda q(R)
	\ge
	\alpha\lambda.
\end{aligned}
\]
By the Chernoff bound and the assumption $\lambda\ge 8\ln 2/\alpha$,
\[
\Pr\left[
Z_R<\frac{\alpha\lambda}{2}
\right]
\le
\exp\left(-\frac{\alpha\lambda}{8}\right)\le \frac12.
\]
Then Claim 1 is proved.
	
Define a sub-event
\[
\mathcal B'
=
\left\{
\exists R\in\mathcal R:\ 
q(R)\ge\alpha,\ 
M\cap R=\emptyset,\ 
{ Z_R \ge \frac{\alpha\lambda}{2}}
\right\}.
\]

{\bf Claim 2.} If $\lambda\ge 8\ln 2/\alpha$, then $\Pr[\mathcal B]\le 2\Pr[\mathcal B'].$

[Proof of Claim 2]  If $\mathcal B$ occurs, then there exists a range $R_M$ with $q(R_M)\ge \alpha$ such
that $M\cap R_M=\emptyset$. By Claim 1, $\Pr\left[Z_{R_M}\ge \frac{\alpha\lambda}{2}\right]\ge
\frac12$. The existence of such an $R_M$ implies that event $\mathcal B'$ occurs. So
\[
\Pr[\mathcal B'\mid \mathcal B]\geq \frac12.
\]
It follows that
\[
\Pr[\mathcal B']=\Pr[\mathcal B]\Pr[\mathcal B'\mid \mathcal B]\geq\frac12\Pr[\mathcal B].
\]
Claim 2 is proved.

Claim 3. If $\lambda\ge 8\ln 2/\alpha$, then $\Pr[\mathcal B']\le\left(\frac{2e\lambda}{d}\right)^d2^{-\alpha\lambda/2}$.

[Proof of Claim 3] Note that $M,T$ can be viewed as two multisets obtained in the following way: for each color class $X_c$, draw $2\lambda_c$ points (with replacement) according to distribution $q/\tau_c$; then randomly assign $\lambda_c$ of them to $M$, the remaining $\lambda_c$ points belong to $T$. So, every sampled point belongs to $T$ with probability $1/2$.

For a range $R$, denote by
\[
A_R=R\cap (M\cup T)
\]
the multiset of points sampled from $R$. Note that $A_R$ is a random multi-set depending on $M$ and $T$. 

Let $E_R$ be the event that $M\cap R=\emptyset$ and $Z_{R}\geq \frac{\alpha\lambda}{2}$. If $E_R$ happens, then $A_R\subseteq T$, and thus $|A_R|=Z_R$. So,
\begin{equation}\label{eq0603-1}
\Pr[E_R]=2^{-|A_R|}=2^{-Z_R}\leq 2^{-\alpha\lambda/2}.
\end{equation}

By the Sauer--Shelah lemma \cite{Sauer1972,Shelah1972}, the number of distinct
traces 
\begin{equation}\label{eq0603-2}
|\mathcal R|_{M\cup T}|\leq \left(\frac{2e\lambda}{d}\right)^d.
\end{equation}

Combining \eqref{eq0603-1} and \eqref{eq0603-2}, taking a union bound over all possible traces in $\mathcal R|_{M\cup T}$, we have
\[
	\Pr[\mathcal B']
	\le
	\left(\frac{2e\lambda}{d}\right)^d
	2^{-\alpha\lambda/2}.
\]
Claim 3 is proved.

Combining Claim 2 and Claim 3,
	\[
	\Pr[\mathcal B]
	\le
	2
	\left(\frac{2e\lambda}{d}\right)^d
	2^{-\alpha\lambda/2}.
	\]
	Choosing $C$ sufficiently large makes the right hand side at most $\varphi$.
	This proves the lemma.
\end{proof}
\begin{theorem}[VC-bounded range spaces]
	\label{thm:vc-stratified}
	Suppose $(X,\mathcal R)$ has VC dimension $d$. Then there exists a
	constant $C>0$ such that for sample size
	\[
	\lambda
	\ge
	\frac{C\Gamma}{\varepsilon}
	\left(
	d\log\frac{\Gamma}{\varepsilon}
	+
	\log\frac1\varphi
	\right),
	\]
	Algorithm~\ref{alg:target-stratified-sampling} returns, with probability
	at least $1-\varphi$, a fair $\varepsilon$-net.
\end{theorem}

\begin{proof}
	By Lemma~\ref{lem:heavy-under-q}, every $\varepsilon$-heavy range
	$R\in\mathcal R$ has
	\[
	q(R)\ge \frac{\varepsilon}{\Gamma}.
	\]
	Set
	\[
	\alpha=\frac{\varepsilon}{\Gamma}.
	\]
	By Lemma~\ref{lem:stratified-vc-eps-net}, with probability at least
	$1-\varphi$, the sampled multiset $M$ hits every range $R$ with
	$q(R)\ge\alpha$. Hence $M$ hits every $\varepsilon$-heavy range.

Because the output $S$ is obtained from $M$ by removing duplicate samples which meets
target-ratio constraints by construction, it is a fair $\varepsilon$-net.
\end{proof}

\subsection{Consequences for DP and Custom Ratios}
\label{Consequences for DP}

Under demographic parity, $\tau_c=\mu_c$ for all colors, so $\Gamma=1$.
Theorem~\ref{thm:vc-stratified} gives the following guarantee.

\begin{corollary}[DP fair $\varepsilon$-net]
	\label{cor:dp-no-logk}
	Suppose $(X,\mathcal R)$ has VC dimension $d$. Under demographic
	parity, there exists a constant $C>0$ such that for sample size
	\[
	\lambda
	\ge
	\frac{C}{\varepsilon}
	\left(
	d\log\frac{1}{\varepsilon}
	+
	\log\frac1\varphi
	\right),
	\]
	the target-stratified sampling returns, with probability at least
	$1-\varphi$, a DP-fair $\varepsilon$-net.
\end{corollary}

The next lemma shows that for custom-ratio fairness, the dependence on $\Gamma$ is unavoidable.

\begin{lemma}[$\Gamma$ dependence is necessary]
	\label{lem:gamma-lower-bound}
	For every $\Gamma\ge 1$, there are range-space instances with
	distribution-shift parameter $\Gamma$ for which every fair $\varepsilon$-net
	satisfying the custom-ratio fairness has size
	\[
	\Omega\left(\frac{\Gamma}{\varepsilon}\right).
	\]
\end{lemma}

\begin{proof}
 Construct an instance in which a color $c_0$ has $\mu_{c_0}=1/2$ and 
	$\tau_{c_0}=1/(2\Gamma)$. Then $\mu_{c_0}/\tau_{c_0}=\Gamma$.  The other ratios $\tau_c$ are chosen arbitrarily subject to the constraint, together with $\tau_{c_0}$, add up to 1.
	
	Inside $X_{c_0}$, construct pairwise disjoint ranges
	\[
	R_1,\ldots,R_L\subseteq X_{c_0},
	\]
	each of size $|R_j|=\varepsilon |X|$, where
	\[
	L=\left\lfloor \frac{1}{2\varepsilon}\right\rfloor
	=
	\Omega\left(\frac1\varepsilon\right).
	\]
	 Such ranges $R_1,\ldots,R_L$ exist because $\mu_{c_0}=1/2$.
	Every $\varepsilon$-net must hit all these ranges. Since they are disjoint
	and contained in $X_{c_0}$, any feasible set $S$ must satisfy
	\[
	|S\cap X_{c_0}|\ge L.
	\]
	If $S$ satisfies $|S\cap X_{c_0}|=\tau_{c_0}|S|$, then
	\[
	|S|
	\ge
	\frac{L}{\tau_{c_0}}
	=
	\Omega\left(\frac{\Gamma}{\varepsilon}\right).
	\]
	The lemma is proved.
\end{proof}

\begin{remark}[Weighted range spaces]
	The same analysis extends to weighted range spaces. Let $\omega_p\ge 0$ be the
	input weight of point $p\in X$, and write
	\[
	\omega(A)=\sum_{p\in A}\omega_p .
	\]
	A range $R$ is weighted $\varepsilon$-heavy if
	\[
	\omega(R)\ge \varepsilon \omega(X).
	\]
	For each color class, define the weighted input proportion
	\[
	\mu_c^w=\frac{\omega(X_c)}{\omega(X)}
	\]
	and set
	\[
	\Gamma_w=\max_{c:\tau_c>0}\frac{\mu_c^w}{\tau_c}.
	\]
	Then the proof of Lemma~\ref{lem:heavy-under-q} gives
	\[
	q_w(R)\ge \frac{\varepsilon}{\Gamma_w}
	\]
	for every weighted $\varepsilon$-heavy range $R$. The sampling bounds
	carry over to weighted range spaces replacing $\Gamma$ with $\Gamma_w$.
\end{remark}

\subsection{Improved Fair Geometric Hitting Set}
\label{sc improve fair ge hit}

The following is an LP relaxation for the fair geometric
hitting set (FGHS) problem:
\[
\begin{array}{ll}
	\min & f \\[1mm]
	\text{s.t.}
	&
	\left\{
	\begin{array}{ll}
		\displaystyle f=\sum_{p\in X}z_p, 
		& \\[2mm]
		
		\displaystyle \sum_{p\in R}z_p\ge 1,
		& \forall R\in\mathcal R,\\[3mm]
		
		\displaystyle \sum_{p\in X_c}z_p=\tau_c f,
		& \forall c\in[k],\\[3mm]
		
		\displaystyle 0\le z_p\le 1,
		& \forall p\in X.
	\end{array}
	\right.
\end{array}
\]

 The algorithm is presented in the following.

\begin{algorithm}[htbp]
	\caption{\textsc{LP-Stratified-FGHS-Rounding}}
	\label{alg:lp-stratified-fghs}
	\begin{minipage}{0.95\linewidth}
		\textbf{Input:} FGHS instance $(X,\mathcal R)$, color partition
		$X=X_1\cup\cdots\cup X_k$, target ratio vector
		$\tau=(\tau_1,\ldots,\tau_k)$, and sample size $\lambda$.
		
		\textbf{Output:} a set $S\subseteq X$.
		
		\begin{enumerate}
			\item Let $(z,f)$ be an optimal solution to the LP.
			Set
			\[
			w_p=\frac{z_p}{f},
			\qquad p\in X.
			\]
			
			\item For each color $c\in[k]$, set $\lambda_c=\tau_c\lambda$.
			
			\item Sample $\lambda_c$ points from $X_c$ with replacement, where
			each point $p\in X_c$ is chosen with probability $w_p/\tau_c$.
			Let $M_c$ be the resulting multiset.
			
			\item Let $S_c$ be the set of distinct points in $M_c$.
			If $|S_c|<\lambda_c$, then add $\lambda_c-|S_c|$ arbitrary points
			from $X_c\setminus S_c$.
			
			\item Return $S=\bigcup_{c=1}^k S_c$.
		\end{enumerate}
	\end{minipage}
\end{algorithm}

\begin{theorem}[Improved Fair Geometric Hitting Set]
	\label{thm:improved-fghs}
	Suppose $(X,\mathcal R)$ has VC dimension $d$. Let
	$OPT_{\mathrm{FGHS}}$ be the size of a minimum fair geometric hitting set.
	Choose
	\[
	\lambda
	=
	O\left(
	f
	\left(
	d\log f+\log\frac1\varphi
	\right)
	\right),
	\]
where $f$ is the optimal objective value of the LP,  then Algorithm~\ref{alg:lp-stratified-fghs} returns, with probability
	at least $1-\varphi$, a fair hitting set $S$ satisfying
	\[
	|S|
	=
	O\left(
	OPT_{\mathrm{FGHS}}
	\left(
	d\log OPT_{\mathrm{FGHS}}
	+
	\log\frac1\varphi
	\right)
	\right).
	\]
\end{theorem}

\begin{proof}
	Let $w$ be the distribution defined in Algorithm~\ref{alg:lp-stratified-fghs}.
	By the LP constraints, we have
	\[
	\sum_{p\in X}w_p=1,
	\quad
	w(X_c)=\tau_c,\ \forall c\in[k],
	\quad
	w(R)\ge \frac1f,\ \forall R\in\mathcal R .
	\]
	Define the family of $1/f$-heavy ranges under $w$ by
	\[
	\mathcal R^{w}_{1/f}
	=
	\left\{
	R\in\mathcal R:\ w(R)\ge \frac1f
	\right\}.
	\]
	The  second constraint of the LP implies
	\[
	\mathcal R\subseteq \mathcal R^{w}_{1/f}.
	\]
	Thus an $1/f$-net hits all ranges in $\mathcal R$.  The sampling
	step of Algorithm~\ref{alg:lp-stratified-fghs} matches the setting of
	Lemma~\ref{lem:stratified-vc-eps-net} with
	\[
	q=w
	\qquad\text{and}\qquad
	\alpha=\frac1f .
	\]
	By the choice of $\lambda$, Lemma~\ref{lem:stratified-vc-eps-net} implies
	that the sampled multiset hits every range in $\mathcal R$ with probability
	at least $1-\varphi$.
	
	The returned set contains the support of the sampled multiset and only adds
	points, so it also hits every range in $\mathcal R$. Moreover, for each color
	$c$, the algorithm returns exactly $\lambda_c=\tau_c\lambda$ points from
	$X_c$. Hence the returned set satisfies the target ratios. Therefore $S$ is a
	fair hitting set with probability at least $1-\varphi$.
	
	 For the approximate effect, consider an optimal fair hitting set $S^*$. Construct variable $z_p^*$ to be the indicator of whether point $p\in S^*$ and let $f^*=\sum_{p\in\mathcal P}z_p^*$. Then $(z^*,f^*)$ satisfies all constraints of the LP. As a consequence, $(z,f)$, being an optimal (fractional) solution to the LP, has
	\[
	f\le |S^*|=OPT_{\mathrm{FGHS}}.
	\]Combining this with $|S|=\lambda$ and the choice of $\lambda$,
	\[
	|S|
	=
	O\left(
	OPT_{\mathrm{FGHS}}
	\left(
	d\log OPT_{\mathrm{FGHS}}
	+
	\log\frac1\varphi
	\right)
	\right).
	\]
	
	The theorem is proved.
	\end{proof}

\subsection{From FGHS to Custom-Ratio Fair $\varepsilon$-Nets}
\label{from fghs to custom}

 Interestingly, we could use the above FGHS algorithm to construct a custom-ratio fair $\varepsilon$-net whose size approximate an optimal one without depending on parameter $\Gamma$.

Recall that a custom-ratio fair $\varepsilon$-net asks for a set of points to hit all ranges in
\[
\mathcal R_\varepsilon
=
\{R\in\mathcal R:\ |R|\ge \varepsilon |X|\},
\]
subject to a target ratio vector constraint $\tau$.

Viewing $(X,\mathcal R_\varepsilon)$ as an FGHS instance, apply Algorithm~\ref{alg:lp-stratified-fghs}. By
Theorem~\ref{thm:improved-fghs}, with probability at least $1-\varphi$, the
output is a custom-ratio fair $\varepsilon$-net for the original instance $(X,\mathcal R)$.
Let $OPT_{\mathrm{CR}\text{-}\varepsilon\mathrm{Net}}$ be the size of a minimum custom-ratio fair $\varepsilon$-net.  By observing that
\[
f\leq OPT_{\mathrm{FGHS}}(X,\mathcal R_\varepsilon)
=
OPT_{\mathrm{CR}\text{-}\varepsilon\mathrm{Net}},
\]
we have the following result.

\begin{corollary}[Custom-ratio fair $\varepsilon$-net via FGHS]
	\label{cor:cr-epsnet-via-fghs}
	Let $d$ be the VC dimension of $(X,\mathcal R)$. There is a randomized algorithm that returns, with probability at least $1-\varphi$, a custom-ratio fair $\varepsilon$-net $S$ satisfying
	\[
	|S|
	=
	O\left(
	OPT_{\mathrm{CR}\text{-}\varepsilon\mathrm{Net}}
	\left(
	d\log OPT_{\mathrm{CR}\text{-}\varepsilon\mathrm{Net}}
	+
	\log\frac1\varphi
	\right)
	\right).
	\]
\end{corollary}

\begin{remark}
		The lower bound above is an absolute-size lower bound, whereas the
		LP-rounding guarantee in Theorem~\ref{thm:improved-fghs} and
		Corollary~\ref{cor:cr-epsnet-via-fghs} is relative to the optimum feasible
		solution. In the lower-bound construction, this optimum is itself
		\[
		\Omega\left(\frac{\Gamma}{\varepsilon}\right).
		\]
		Hence the LP-rounding guarantee does not contradict the lower bound.
\end{remark}

\section{Experimental Evaluation}
\label{sec:experiments}

We evaluate the target-stratified framework from several aspects.
We first test its performance on real datasets under demographic parity
and custom-ratio fairness, including settings with many intersectional groups.
We then use synthetic datasets to study the effects of dataset size, the number
of protected groups, the distribution-shift parameter $\Gamma$, and group
overlap. We also compare our method with Fair Sketch-and-Merge (FSM), evaluate
our rounding method for Fair Geometric Hitting Set (FGHS), and study the use of
fair summaries for downstream range-query filtering.

\subsection{Experimental Setup}

\paragraph{Datasets and range workloads.}
We use two real tabular datasets:
\href{https://archive.ics.uci.edu/dataset/2/adult}{Adult} and
\href{https://www.propublica.org/datastore/dataset/compas-recidivism-risk-score-data-and-analysis}{COMPAS}.
Each record is treated as a point in a numerical feature space, and range
queries are axis-aligned hyperrectangles. A summary is valid if it satisfies
the required group quotas and hits every tested $\varepsilon$-heavy range.

\begin{table}[htbp]
	\centering
	\caption{Dataset settings.}
	\label{tab:dataset-settings}
	\footnotesize
	\renewcommand{\arraystretch}{1.15}
	\setlength{\tabcolsep}{3pt}
	\begin{tabularx}{\columnwidth}{@{}lclX@{}}
		\toprule
		Dataset & \# Records & Group & Range attributes \\
		\midrule
		
		Adult
		& 32,561
		& \texttt{sex}
		& \texttt{age}, \texttt{education\_num},
		\texttt{hours\_per\_week}, \texttt{capital\_gain},
		\texttt{capital\_loss}, \texttt{fnlwgt} \\
		
		COMPAS
		& 7,214
		& \texttt{race}
		& \texttt{age}, \texttt{priors\_count},
		\texttt{decile\_score}, \texttt{juv\_fel\_count},
		\texttt{juv\_misd\_count}, \texttt{juv\_other\_count} \\
		
		\bottomrule
	\end{tabularx}
\end{table}

Preprocessing keeps all records. Missing group values are mapped to
\texttt{Unknown}, and missing numerical values are replaced by the median.
For each dataset, we generate quantile-box workloads using five empirical
quantile thresholds for each numerical attribute. Empty and duplicate ranges
are removed, and the workload contains at most $20{,}000$ ranges.

\paragraph{Compared methods.}
Our main baseline is the sample-and-repair method of
Dehghankar et al.~\cite{dehghankar2025fair}.
It first draws a uniform sample and then adds records from underrepresented
groups until the required group ratios are satisfied.
Our method instead fixes the number of samples from each group before
sampling.

Both methods use the same fairness constraints and the same heavy-range
validation procedure. Our implementation builds on the code released by
Dehghankar et al.~\cite{fairnetcode}. The code and data used in our experiments
are available at
\href{https://github.com/artifact-hub/code}{GitHub}.
We also compare with Fair Sketch-and-Merge (FSM) later in this section.

\paragraph{Fairness settings.}
We consider demographic parity (DP) and custom-ratio fairness (CR).
Under DP, the target ratios are the empirical group proportions in the full
dataset. Under CR, the target ratios are specified in
Table~\ref{tab:cr-target-ratios}. When the target counts are not integral, we
use the largest-remainder rule to obtain integer quotas.

\begin{table}[htbp]
	\centering
	\caption{Custom target ratios.}
	\label{tab:cr-target-ratios}
	\scriptsize
	\renewcommand{\arraystretch}{1.15}
	\setlength{\tabcolsep}{3pt}
	\begin{tabularx}{\columnwidth}{@{}lX@{}}
		\toprule
		Dataset & Custom target ratios \\
		\midrule
		
		Adult
		& Female: $0.50$, Male: $0.50$ \\
		
		COMPAS
		& African-American: $0.40$, Caucasian: $0.35$,
		Hispanic: $0.11$, Other: $0.09$, Asian: $0.025$,
		Native American: $0.015$ \\
		
		\bottomrule
	\end{tabularx}
\end{table}

\paragraph{Evaluation metrics.}
For the real-data DP and CR experiments, we use
\[
\varepsilon\in\{0.05,0.08,0.10,0.15,0.20\}.
\]
For each dataset, fairness setting, and value of $\varepsilon$, we run
$30$ independent trials with at most $200$ attempts per trial.
We report the average size of a valid summary and the average time needed to
obtain the first valid summary. The running time includes sampling, repair when
needed, and validation against all tested heavy ranges. All reported outputs
are checked for both fairness and range coverage.

\subsection{Fair Range Summaries on Real Data}

We first evaluate the two methods on Adult and COMPAS. We start with DP and CR
constraints based on one protected attribute and then consider intersectional
groups formed by multiple attributes.

\subsubsection{DP and CR Fairness}

\begin{figure}[htbp]
	\centering
	\includegraphics[width=\linewidth]
	{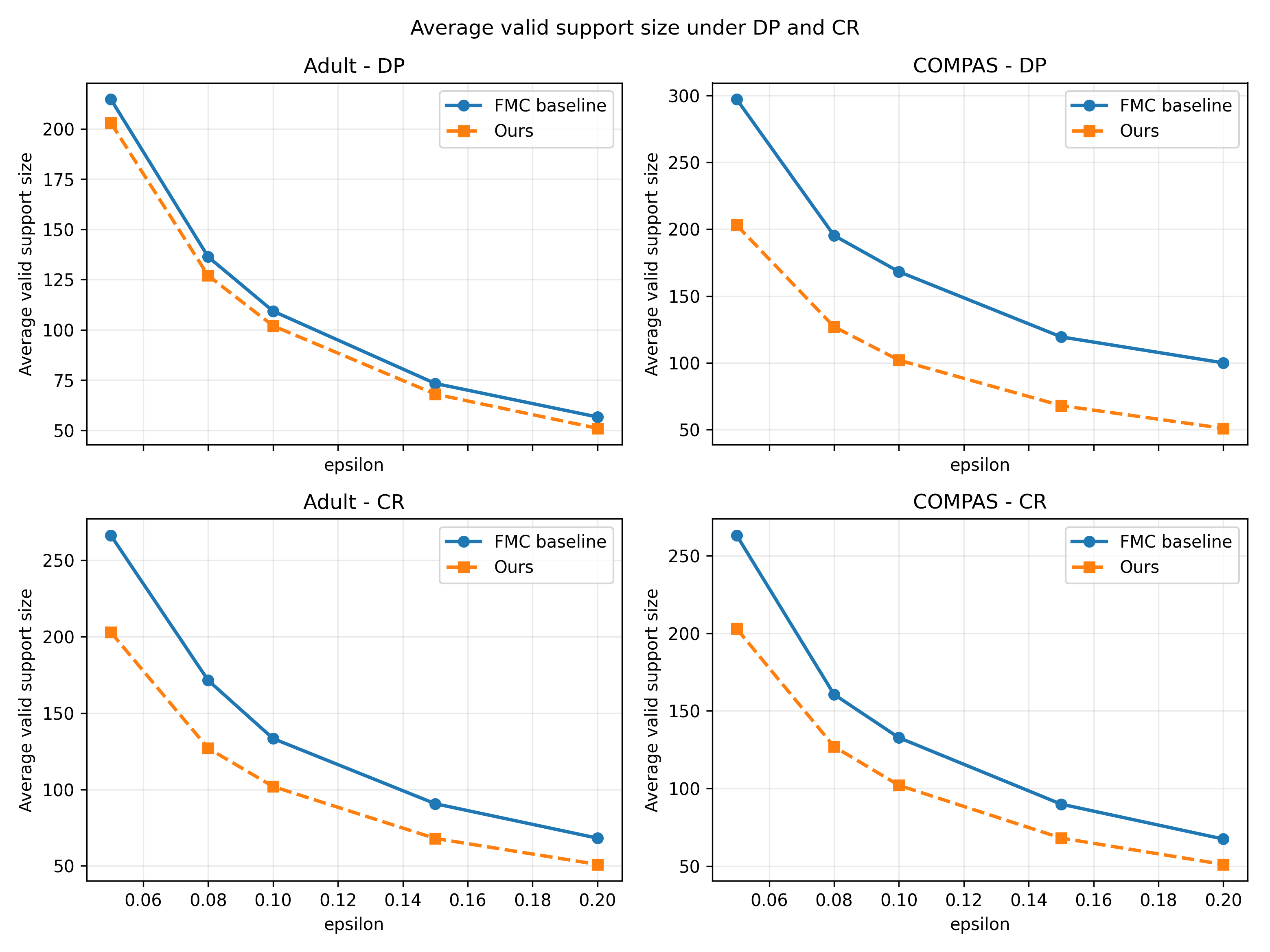}
	\caption{Average valid summary size under DP and CR target ratios.}
	\label{fig:dp-cr-sample-size}
\end{figure}

\begin{figure}[htbp]
	\centering
	\includegraphics[width=\linewidth]
	{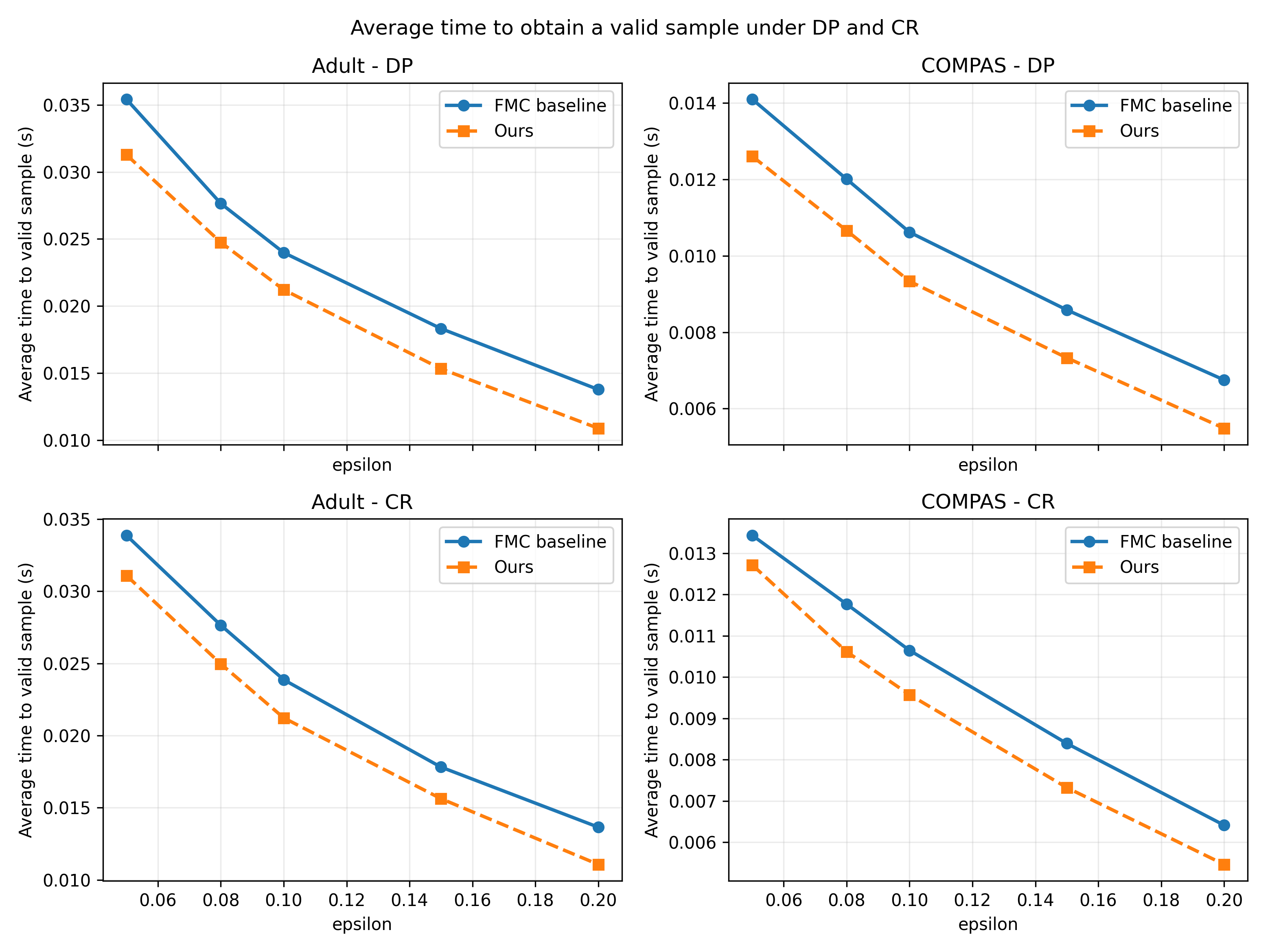}
	\caption{Average time to obtain the first valid summary under DP and CR
		target ratios.}
	\label{fig:dp-cr-time}
\end{figure}

Figures~\ref{fig:dp-cr-sample-size} and~\ref{fig:dp-cr-time} show that our
method consistently produces smaller valid summaries than sample-and-repair.

Under CR, our method reduces the summary size by
$23.5\%$--$25.9\%$ on Adult and $20.9\%$--$24.3\%$ on COMPAS.
Under DP, the reduction is $5.5\%$--$10.0\%$ on Adult and
$31.7\%$--$49.0\%$ on COMPAS.
Our method is also faster in these experiments.

The main reason is that sample-and-repair may need to add records after the
initial sample is drawn. Our method fixes the group quotas before sampling and
therefore avoids this repair step.

\subsubsection{Intersectional Group Scaling}

The previous experiments use one protected attribute. We next consider
fairness constraints defined by several attributes, which can produce many
intersectional groups.

We fix $\varepsilon=0.10$. For Adult, we start with \texttt{sex} and
successively add \texttt{race}, \texttt{education}, \texttt{workclass}, and
\texttt{marital\_status}. For COMPAS, we start with \texttt{race} and add
\texttt{sex}, \texttt{age}, charge degree, and prior-count buckets.
Groups with fewer than $50$ records are merged into \texttt{Other}.

The initial sample size $\lambda$ is the same as in the DP experiment. It is
increased only when necessary to assign at least one point to every group with
a positive target ratio.

\begin{figure}[htbp]
	\centering
	\includegraphics[width=0.48\linewidth]
	{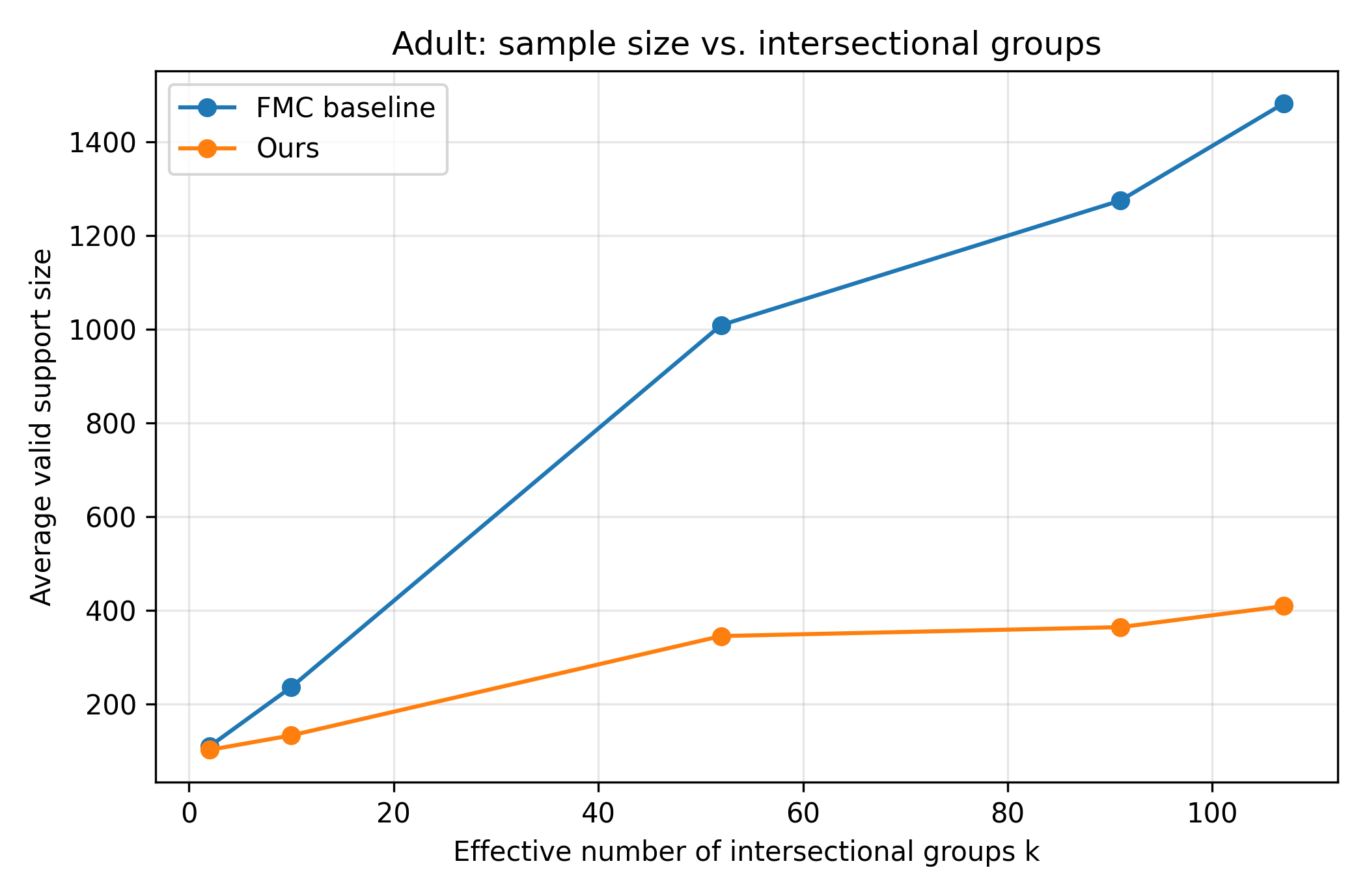}
	\hfill
	\includegraphics[width=0.48\linewidth]
	{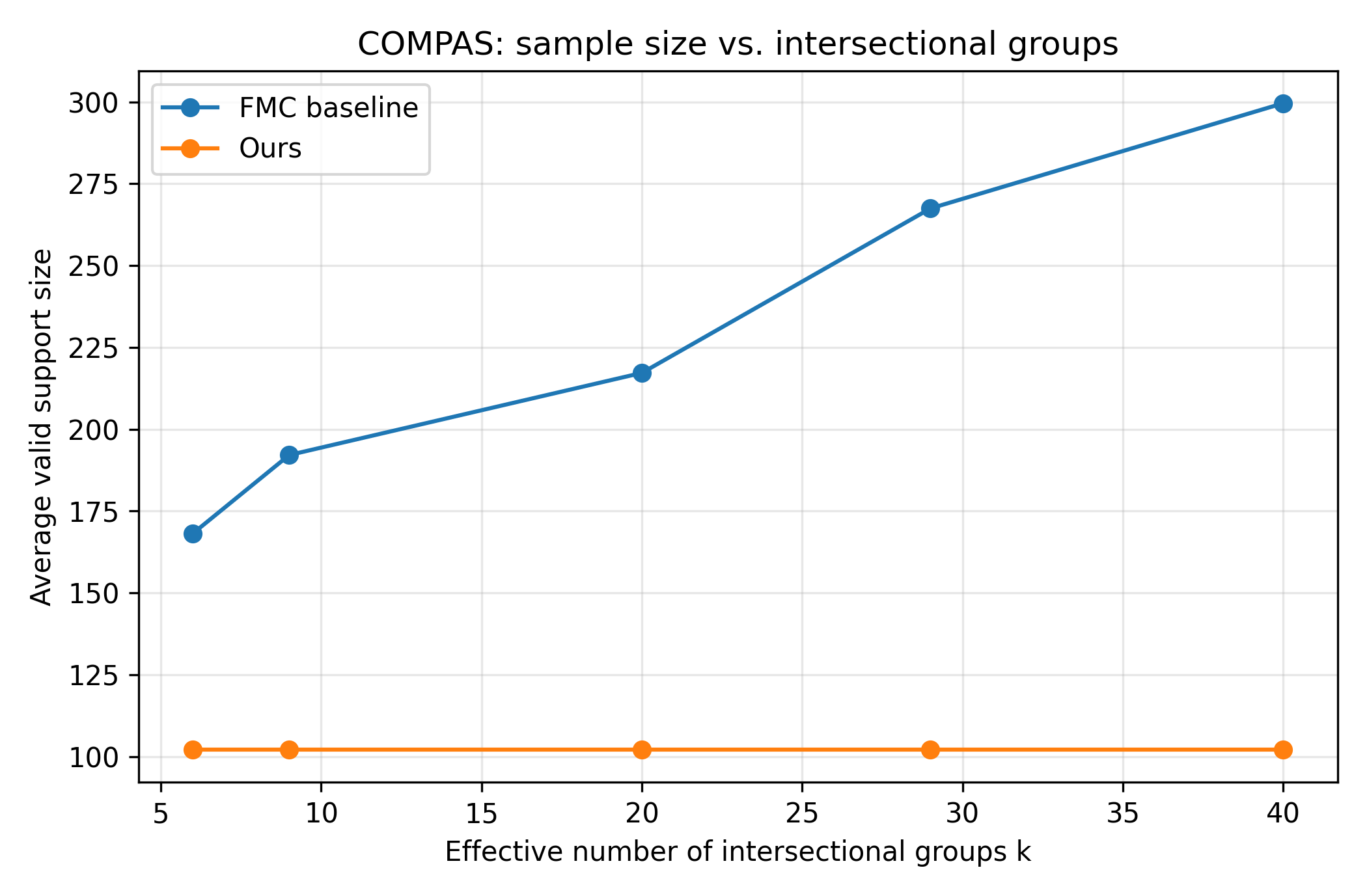}
	\caption{Valid summary size as the number of intersectional groups $k$
		increases. Left: Adult. Right: COMPAS.}
	\label{fig:k-scaling}
\end{figure}

Figure~\ref{fig:k-scaling} shows that the difference between the two methods
becomes much larger as the group partition becomes finer.

On Adult, the number of groups increases from $2$ to $107$.
The average summary size of sample-and-repair increases from $109.3$ to
$1482.2$, while that of our method increases from $102$ to $409$.

On COMPAS, the number of groups increases from $6$ to $40$.
The summary size of sample-and-repair increases from $168.2$ to $299.6$,
while our method remains at $102$. In this case, the initial sample size is
already large enough to assign a positive quota to every group.

These results show that post-hoc repair becomes costly when the number of
groups is large. Our method fixes all group quotas before sampling and avoids
most of this increase.

\subsection{Scalability and Sensitivity Analysis}

We next use synthetic datasets to study individual factors separately.
We consider the dataset size, the number of protected groups, the parameter
$\Gamma$, and the overlap between groups.

\subsubsection{Dataset Size}

We first study the effect of the dataset size. We fix $\varepsilon=0.1$ and
the other main settings, and increase the number of records.

\begin{figure}[htbp]
	\centering
	\includegraphics[width=\linewidth]
	{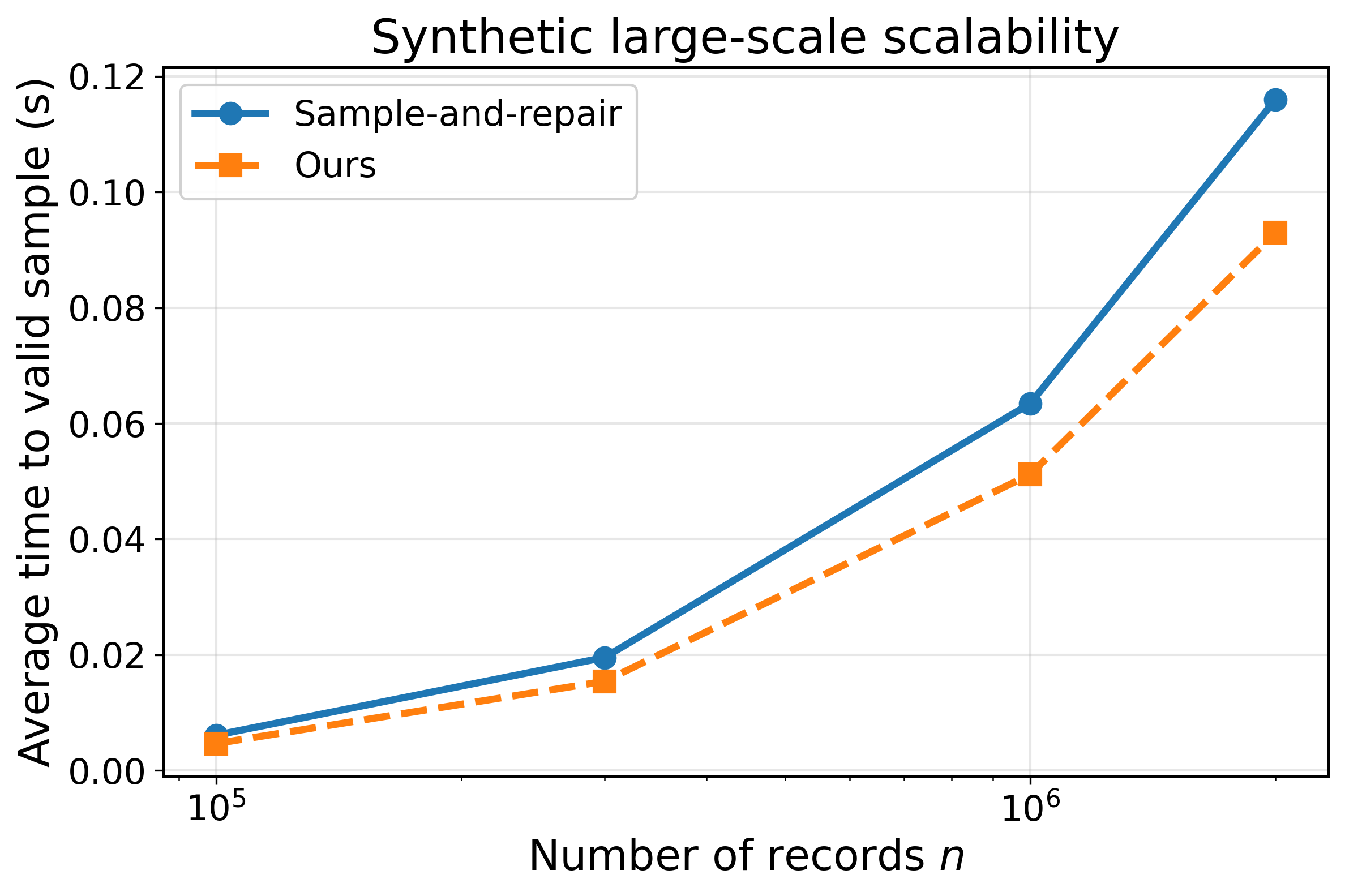}
	\caption{Running time as the dataset size increases.}
	\label{fig:synthetic-large-scale-n}
\end{figure}

Figure~\ref{fig:synthetic-large-scale-n} shows that both methods scale well as
the dataset size increases, including datasets with up to two million records.
Our method is slightly faster over the tested range because it does not need
the additional repair step.

This experiment shows that target-stratified sampling remains efficient on
large datasets.

\subsubsection{Number of Protected Groups}

The intersectional experiment above uses group partitions from real data.
Here we vary the number of groups directly while keeping the other main
settings fixed.

\begin{figure}[htbp]
	\centering
	\includegraphics[width=\linewidth]
	{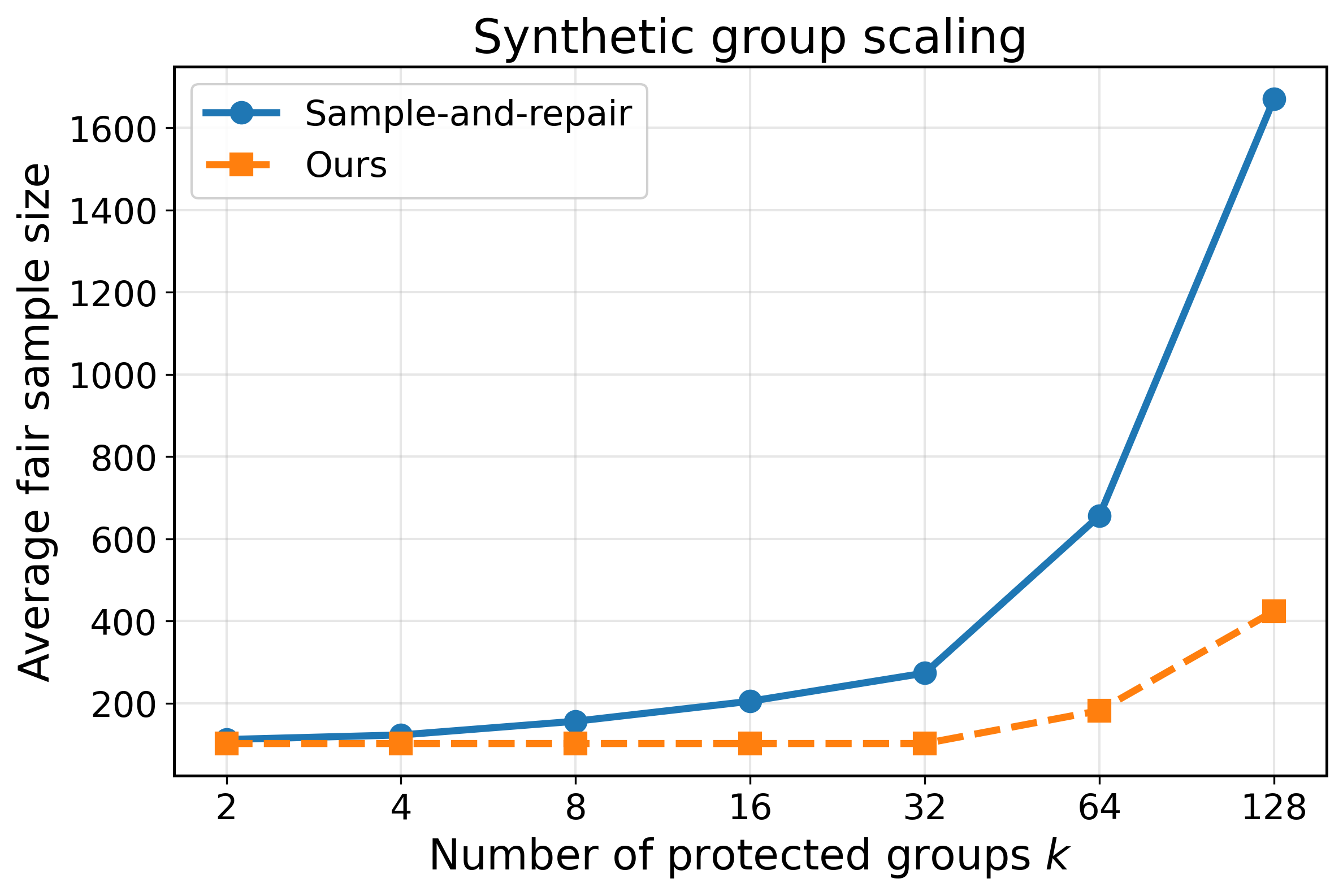}
	\caption{Valid summary size as the number of protected groups increases.}
	\label{fig:synthetic-k-scaling}
\end{figure}

Figure~\ref{fig:synthetic-k-scaling} shows the same trend as the real-data
intersectional experiment. As $k$ increases, the summary size of
sample-and-repair grows much faster than that of our method.

At $k=128$, the average summary size is $1670$ for sample-and-repair and
$424$ for our method.

The size of our summary also increases when $k$ becomes very large, because
every group with a positive target ratio must receive at least one point.
However, the increase is much smaller than that of sample-and-repair.

\subsubsection{Effect of the Distribution-Shift Parameter $\Gamma$}

For custom-ratio fairness, our theoretical bound depends on
\[
\Gamma=
\max_{c:\tau_c>0}\frac{\mu_c}{\tau_c}.
\]
We therefore study how the required summary size changes with $\Gamma$.

We consider
\[
\Gamma\in\{1,2,4,8\}
\]
while keeping $\varepsilon$ fixed.

\begin{figure}[htbp]
	\centering
	\includegraphics[width=\linewidth]
	{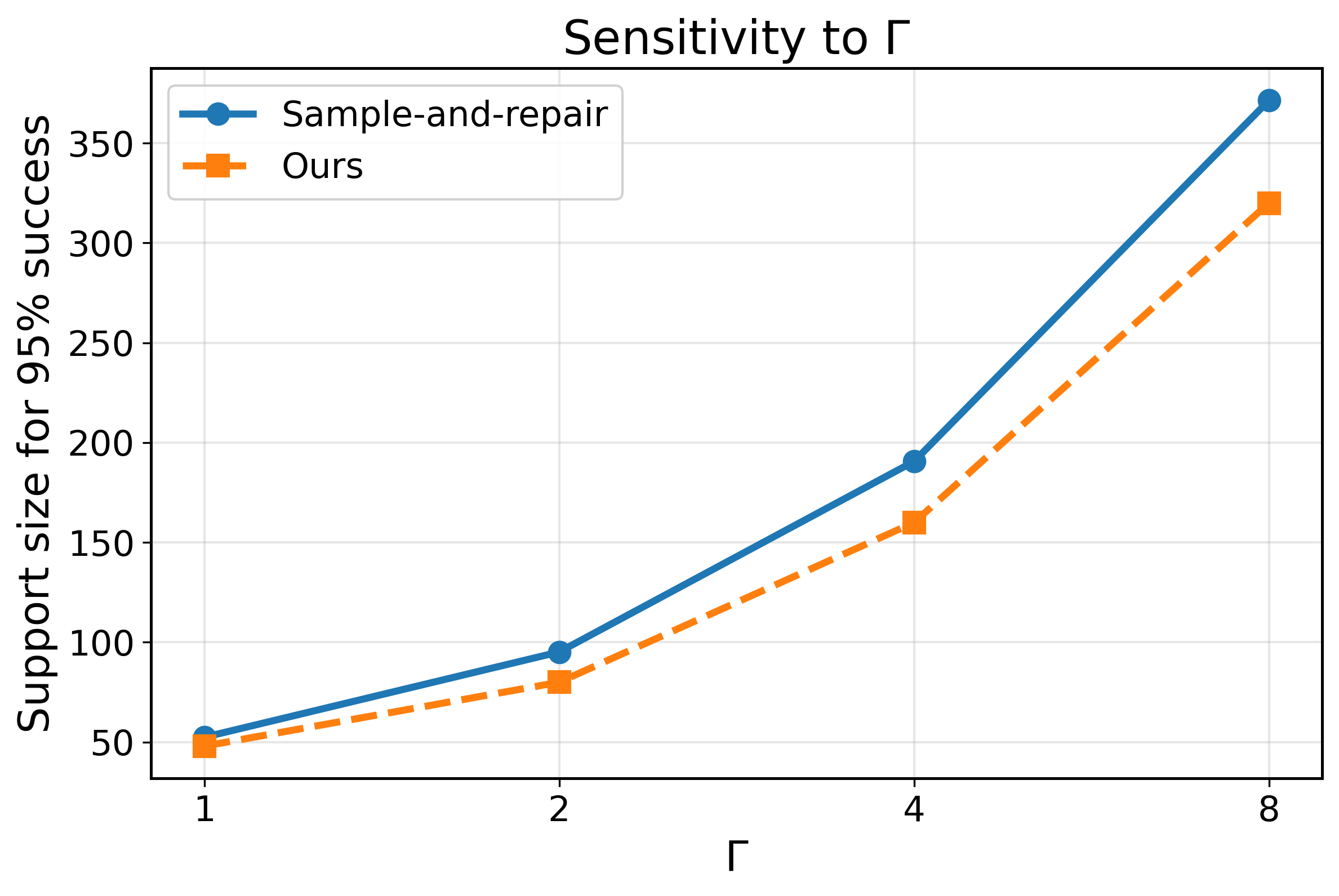}
	\caption{Valid summary size for different values of $\Gamma$.}
	\label{fig:gamma-sensitivity}
\end{figure}

At the $95\%$ empirical success threshold, the summary sizes of our method are
$48$, $80$, $160$, and $320$ for
$\Gamma=1,2,4,8$, respectively. The corresponding sizes for
sample-and-repair are $52$, $95$, $191$, and $371$.

The required summary size increases with $\Gamma$, which agrees with the
$\Gamma/\varepsilon$ dependence in our theoretical analysis.
Our method also uses fewer points than sample-and-repair for every tested value
of $\Gamma$.

\subsubsection{Effect of Group Overlap}

We also test whether the performance depends on the spatial distribution of
the protected groups. We fix
$n=50{,}000$, $k=16$, and $\varepsilon=0.1$, together with the group
proportions and target ratios. We then vary a separation parameter $\rho$
from $0$ to $1$.

When $\rho=0$, the groups have similar feature distributions. As $\rho$
increases, the groups become more separated. Thus, a larger $\rho$ means less
overlap between the groups.

\begin{figure}[htbp]
	\centering
	\includegraphics[width=\linewidth]
	{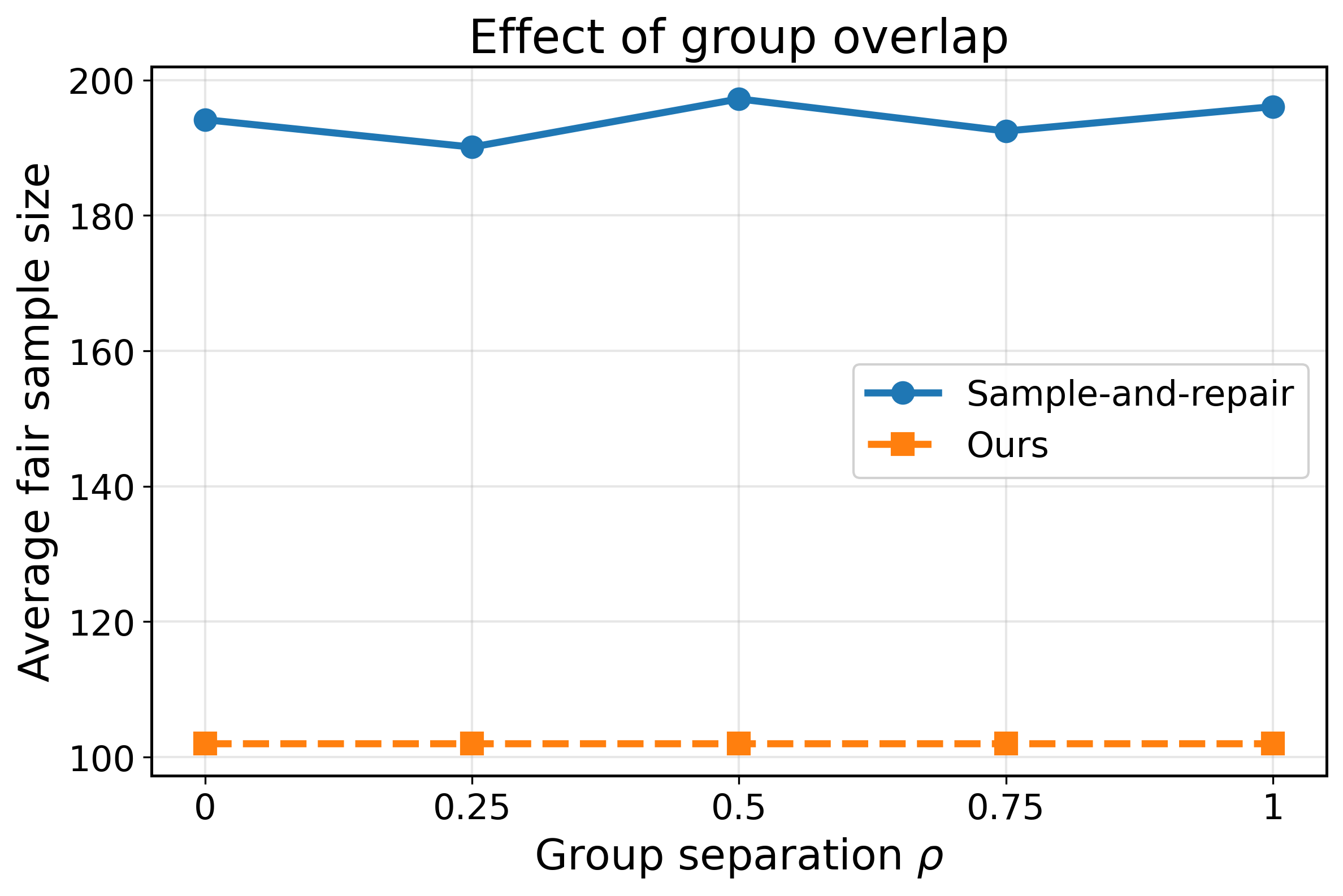}
	\caption{Valid summary size under different levels of group overlap.}
	\label{fig:distribution-overlap}
\end{figure}

Figure~\ref{fig:distribution-overlap} shows that our result is stable over the
tested values of $\rho$. Our summary size remains $102$, while
sample-and-repair uses between $190$ and $197$ points.
Both methods satisfy the tested fairness and $\varepsilon$-net conditions.

Thus, the advantage of our method is not limited to one particular spatial
distribution of the protected groups.

\subsection{Comparison with Fair Sketch-and-Merge}
\label{subsec:fsm-exp}

Sample-and-repair can be used in both DP and CR settings. To compare with
another method from the previous fair $\varepsilon$-net framework, we also
evaluate Fair Sketch-and-Merge (FSM)~\cite{dehghankar2025fair} under DP.

We use an implementation following the released FairNet code.
We fix
\[
n=65{,}536,\qquad
\varepsilon=0.1,\qquad
\lambda=102,
\]
and vary
\[
k\in\{2,4,8,16\}.
\]

For different values of $k$, we use the same data points, rectangles, and
heavy ranges, and only change the group labels. Every output is checked for
both fairness and range coverage. We use the same largest-remainder rule for
non-integral target counts.

\begin{figure}[htbp]
	\centering
	
	\includegraphics[width=0.48\linewidth]
	{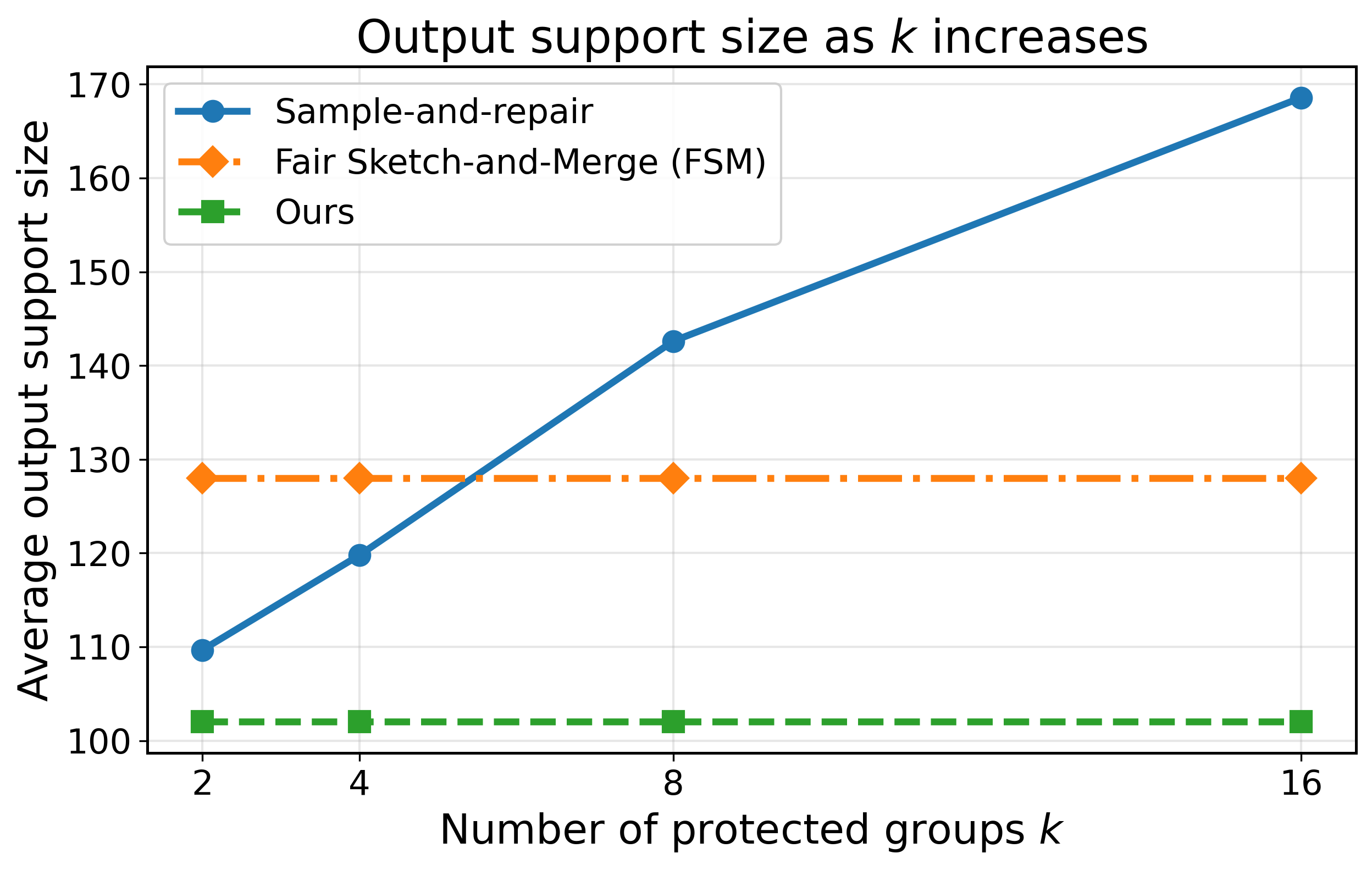}
	\hfill
	\includegraphics[width=0.48\linewidth]
	{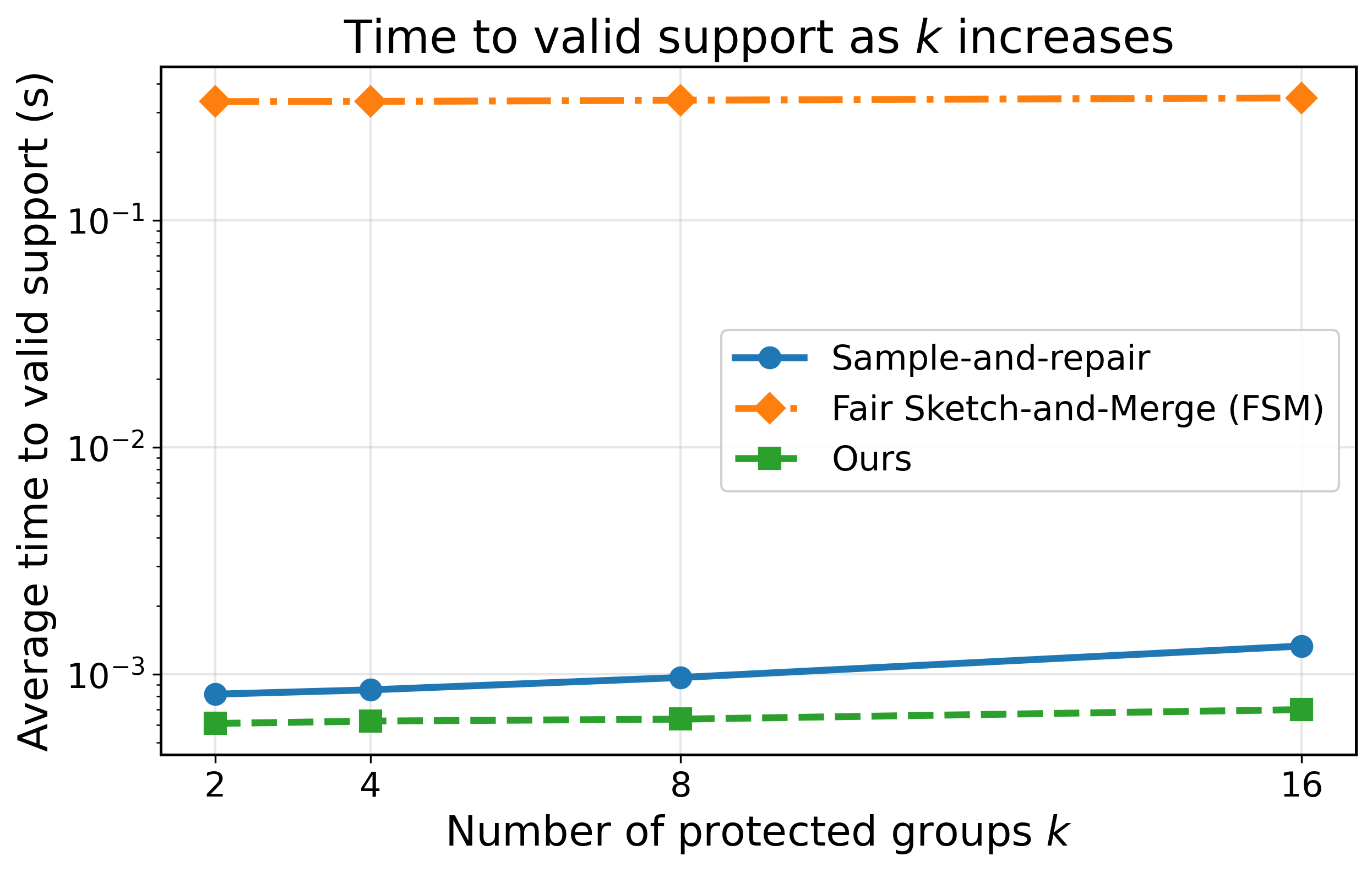}
	
	\caption{Comparison with Fair Sketch-and-Merge as the number of protected
		groups increases. Left: final summary size. Right: running time.}
	\label{fig:fsm-comparison}
\end{figure}

Figure~\ref{fig:fsm-comparison} compares the three methods.
Our method returns $102$ points for every tested value of $k$.
FSM returns $128$ points in all tested settings, while the summary size of
sample-and-repair increases from $110$ at $k=2$ to $169$ at $k=16$.

Our method therefore gives the smallest summary in these experiments.
FSM is also slower in our implementation because it uses fair matching and
repeated halving.

\subsection{Fair Geometric Hitting Set}

We next evaluate our rounding method for Fair Geometric Hitting Set.
For each instance, the two rounding methods start from the same optimal
fractional LP solution. They differ only in the rounding step.

\begin{figure}[htbp]
	\centering
	\includegraphics[width=\linewidth]
	{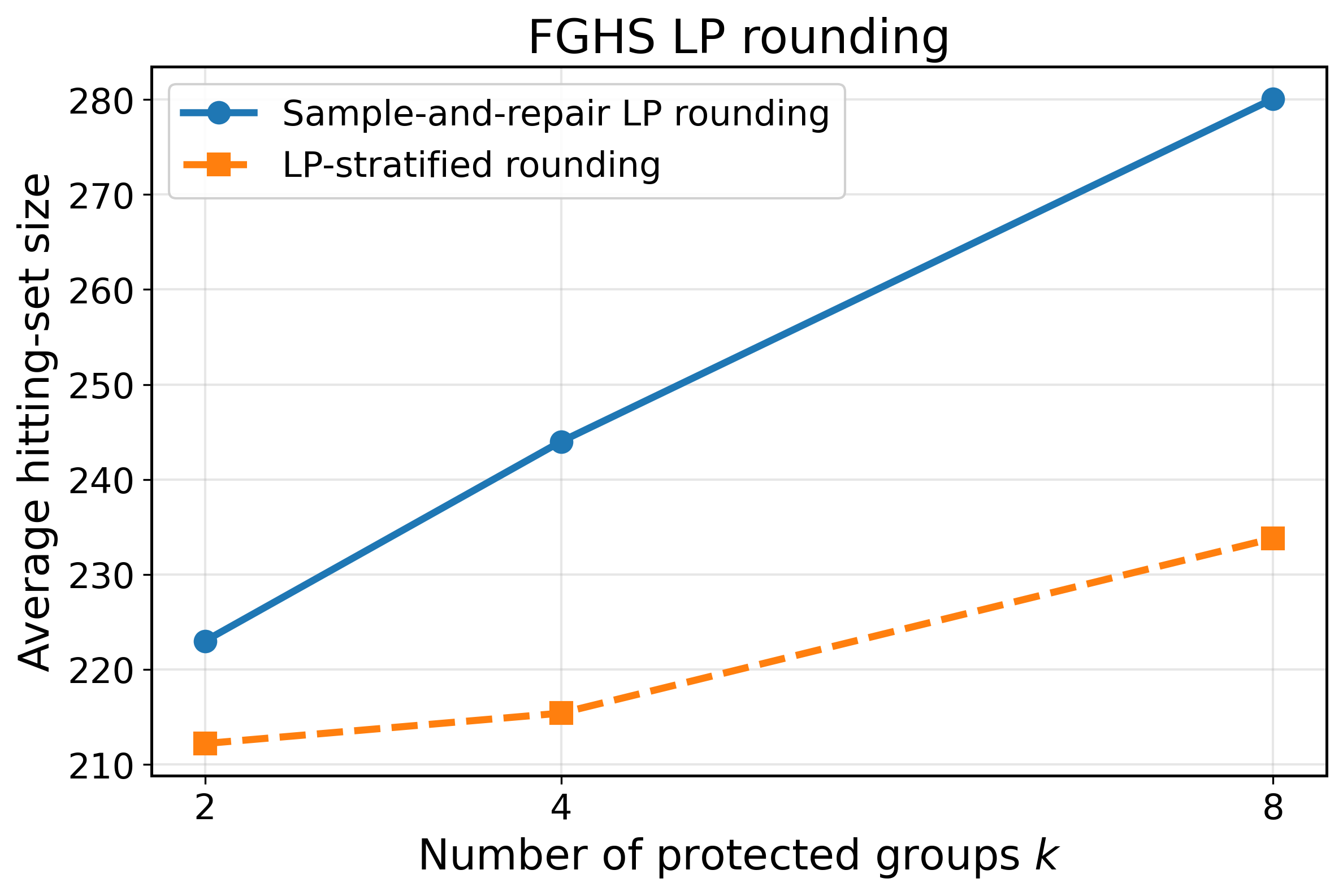}
	\caption{Final hitting-set size for the two FGHS rounding methods.}
	\label{fig:fghs-rounding}
\end{figure}

Both methods satisfy the required fairness constraints and hit all ranges in
the tested instances. Figure~\ref{fig:fghs-rounding} shows that our method
produces smaller hitting sets, and the difference becomes larger as the number
of protected groups increases.

At $k=8$, the average hitting-set size is $280$ for the baseline rounding
method and $234$ for our target-stratified rounding method. This is a
reduction of about $16.5\%$.

Since both methods use the same fractional LP solution, the difference comes
from the rounding step. The result shows that target-stratified rounding can
produce a smaller fair hitting set while keeping the same fairness and range
coverage requirements.

\subsection{Downstream Range-Query Application}

Finally, we study how fair $\varepsilon$-net summaries can be used to filter
range queries.

Let $S$ be an $\varepsilon$-net and let $R$ be a query range. By definition,
\[
S\cap R=\emptyset
\quad\Longrightarrow\quad
|R|<\varepsilon|X|.
\]
Thus, if $S$ is an $\varepsilon$-net for the query family, a query that does
not intersect $S$ has selectivity below $\varepsilon$. Such queries can be
filtered before their exact support is computed on the full dataset.

Queries that intersect the summary are kept for exact evaluation. A smaller
summary may therefore filter more low-selectivity queries while still
preserving sufficiently large ranges.

We test this use case on Adult and COMPAS with a separate set of range queries
that is not used during summary construction.

\begin{figure}[htbp]
	\centering
	\includegraphics[width=\linewidth]
	{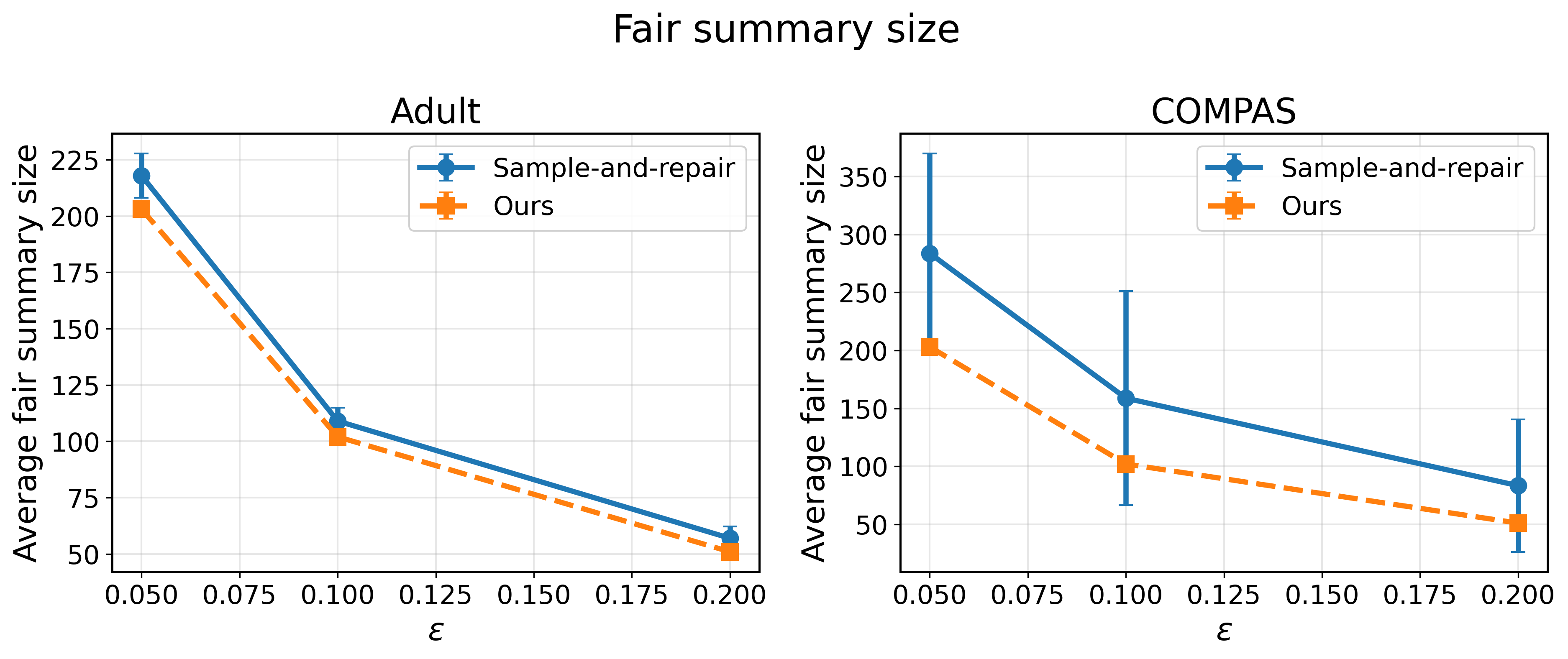}
	\caption{Fair summary size in the downstream range-query experiment.}
	\label{fig:query-summary-size}
\end{figure}

Figure~\ref{fig:query-summary-size} shows that our method produces a smaller
summary for every tested value of $\varepsilon$.
For example, on COMPAS with $\varepsilon=0.10$, the average summary size
decreases from $159$ for sample-and-repair to $102$ for our method, a
reduction of $35.7\%$.

\begin{figure}[htbp]
	\centering
	\includegraphics[width=\linewidth]
	{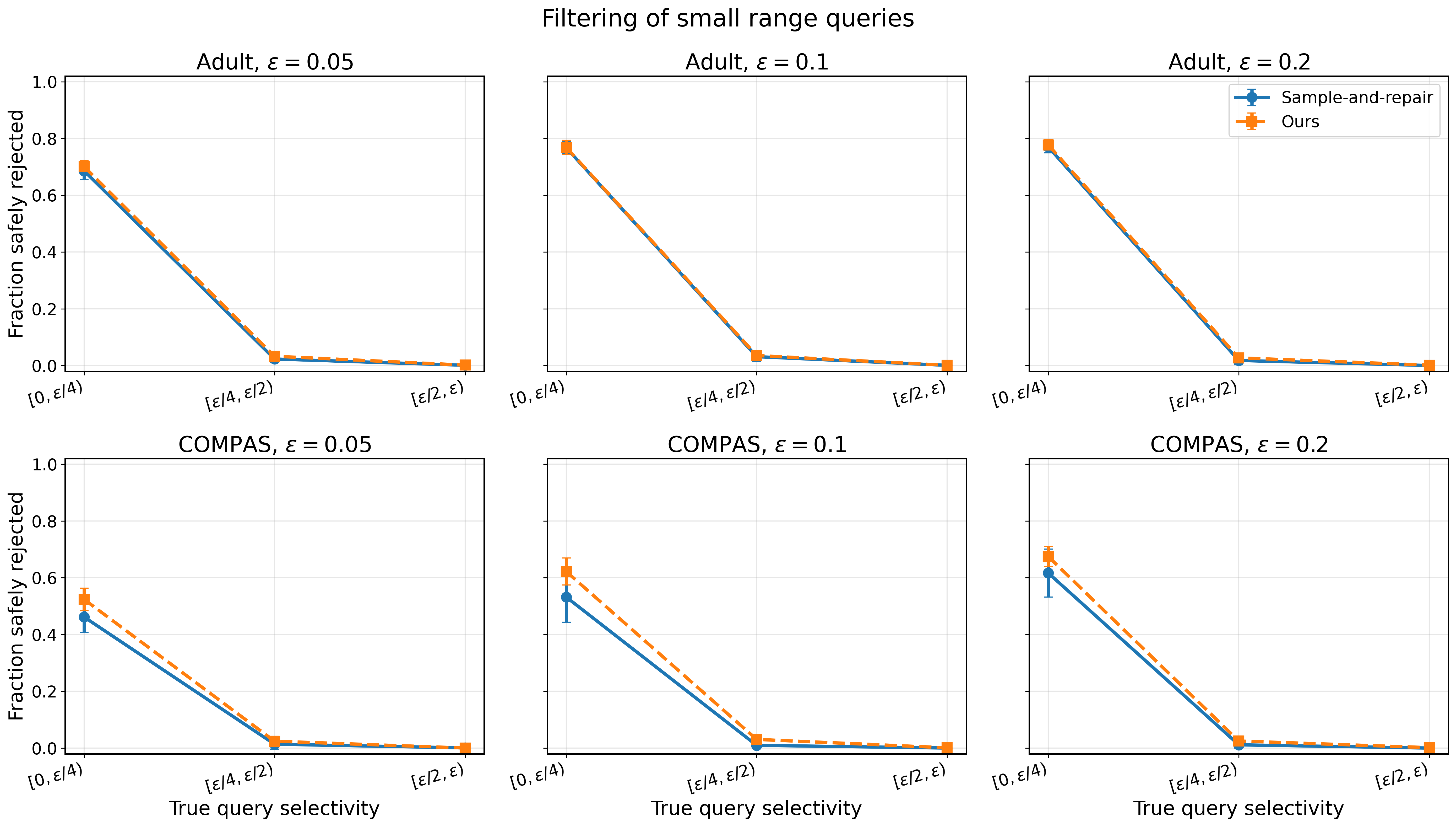}
	\caption{Fraction of low-selectivity queries filtered by the fair summaries.
		The three groups correspond to queries whose true selectivity lies in
		$[0,\varepsilon/4)$, $[\varepsilon/4,\varepsilon/2)$, and
		$[\varepsilon/2,\varepsilon)$, respectively.}
	\label{fig:safe-rejection}
\end{figure}

Figure~\ref{fig:safe-rejection} shows that queries with smaller true
selectivity are more likely to be filtered. Queries whose selectivity is close
to $\varepsilon$ are more likely to intersect the summary and are therefore
kept for exact evaluation.

On COMPAS with $\varepsilon=0.10$, among queries whose true selectivity is
below $\varepsilon/4$, sample-and-repair filters $53.2\%$ of the queries,
while our method filters $62.2\%$. This is an increase of $9.0$ percentage
points.

At the same time, across all tested settings, at least $99.99\%$ of the
queries whose true selectivity is at least $\varepsilon$ are retained.

These results show that the smaller summaries produced by our method can
filter more low-selectivity queries while preserving almost all queries above
the target threshold.

\section{Conclusion}
\label{sec:conclusion}

In this paper, we studied fair $\varepsilon$-nets and fair geometric hitting
sets as fairness-aware range summaries. Unlike the previous
sample-and-repair approach~\cite{dehghankar2025fair}, our method determines
the sample size for each color before sampling. The resulting
target-stratified sampling framework recovers the standard
$\varepsilon$-net guarantee under demographic parity, characterizes the
additional difficulty of custom-ratio fairness through the distribution-shift
parameter $\Gamma$, and improves the approximation guarantee for fair
geometric hitting sets.

Experiments on real and synthetic datasets show that our method constructs
smaller valid fair summaries than existing approaches, remains efficient on
large datasets and fine-grained group partitions, and exhibits the predicted
dependence on $\Gamma$. The experiments also demonstrate the benefit of our
rounding method for fair geometric hitting sets and show that the resulting
smaller fair summaries can improve downstream range-query filtering while
retaining almost all sufficiently large ranges.

Future work includes dynamic and streaming settings, other geometric range
families, and broader fairness constraints.
\section*{Acknowledgment}
This research is supported in part by National Natural Science Foundation of China
(grant number 12571342).

\bibliographystyle{IEEEtran}
\bibliography{reference2}

\end{document}